\documentclass[onefignum,onetabnum]{siamart171218}

\usepackage[letterpaper,centering]{geometry}
\usepackage[textsize=tiny]{todonotes}
\usepackage{mathtools}

\usepackage{xr}
\usepackage{lipsum}
\usepackage{amsfonts}
\usepackage{graphicx}
\usepackage{epstopdf}
\usepackage[textsize=tiny]{todonotes}
\usepackage{algorithmic}
\usepackage{mathtools}
\usepackage{changes}
\colorlet{Changes@Color}{black}
\ifpdf
  \DeclareGraphicsExtensions{.eps,.pdf,.png,.jpg}
\else
  \DeclareGraphicsExtensions{.eps}
\fi

\newsiamthm{remark}{Remark}
\newsiamthm{example}{Example}
\newsiamthm{problem}{Problem}
\crefname{problem}{Problem}{Problems}
\newsiamthm{hypothesis}{Hypothesis}
\crefname{hypothesis}{Hypothesis}{Hypotheses}
\newsiamthm{claim}{Claim}

\headers{Geodesically parameterized covariance estimation}{Musolas, Smith, Marzouk}

\title{Geodesically parameterized covariance estimation\thanks{This work was supported by ``la Caixa'' Banking Foundation (ID 100010434) under project LCF/BQ/AN13/10280009, and by the Air Force Office of Scientific Research Computational Mathematics Program.}}

\author{
  Antoni Musolas\thanks{Center for Computational Science and Engineering, Massachusetts Institute of Technology, Cambridge, MA (\email{musolas@mit.edu},
    \email{ymarz@mit.edu}).}
\and
Steven T.\ Smith\thanks{MIT Lincoln Laboratory, Lexington, MA. Approved for public release. Distribution is unlimited. This material is based upon work supported by the Under Secretary of Defense for Research and Engineering under Air Force Contract No.~FA8702-15-D-0001. Any opinions, findings, conclusions or recommendations expressed in this material are those of the author(s) and do not necessarily reflect the views of the Under Secretary of Defense for Research and Engineering (\email{stsmith@ll.mit.edu}).}
  \and
  Youssef Marzouk\footnotemark[2]
}
\usepackage{amsopn}
\DeclareMathOperator{\diag}{diag}

\ifpdf
\hypersetup{
  pdftitle={GEODESICALLY PARAMETERIZED COVARIANCE ESTIMATION},
  pdfauthor={A.~Musolas, S.~T.~Smith, and Y.~Marzouk}
}
\fi

\colorlet{Changes@Color}{black}

\DeclarePairedDelimiterX{\infdivx}[2]{(}{)}{%
  #1\;\delimsize\|\;#2%
}
\newcommand{\infdiv}{D_{KL}\infdivx}

\DeclareMathOperator*{\argmin}{argmin}
\DeclareMathOperator*{\argmax}{argmax}
\DeclareMathOperator{\Tr}{tr}
\DeclareMathOperator{\Id}{Id}
\DeclareMathOperator{\Sym}{\mathbf{S}(\it{n})}
\DeclareMathOperator{\Sp}{\mathbf{S_+}(\it{n})}
\begin{document}

\maketitle

\begin{abstract}
Statistical modeling of spatiotemporal phenomena often requires selecting a covariance matrix from a covariance class. Yet standard parametric covariance families can be insufficiently flexible for practical applications, while non-parametric approaches may not easily allow certain kinds of prior knowledge to be incorporated. We propose instead to build covariance families out of geodesic curves. These covariances offer more flexibility for problem-specific tailoring than classical parametric families, and are preferable to simple convex combinations. Once the covariance family has been chosen, one typically needs to select a representative member by solving an optimization problem, e.g., by maximizing the likelihood of a data set. We consider instead a differential geometric interpretation of this problem: minimizing the geodesic distance to a sample covariance matrix (``natural projection''). Our approach is consistent with the notion of distance employed to build the covariance family and does not require assuming a particular probability distribution for the data. We show that natural projection and {maximum likelihood estimation within the covariance family} are locally equivalent up to second order. We also demonstrate that natural projection may yield more accurate estimates with noise-corrupted data.
\end{abstract}



\section{Introduction}\label{sec:intro}


Statistical modeling of spatiotemporal phenomena often requires employing and estimating covariance matrices. Classical parametric covariance families (e.g., based on Mat\'ern \cite{cressie1992statistics,rasmussen2006gaussian} kernels) can be insufficiently flexible for practical applications. 
By construction, these approaches describe a high-dimensional object (a symmetric positive semi-definite matrix with $O(n^2)$ degrees of freedom) using only a few generic parameters that are not problem-specific; more broadly, these parametric families may not be rich enough to capture the phenomena of interest.
Non-parametric methods (e.g., sparse precision matrix estimation \cite{cai2011adaptive,friedman2008sparse,cai2011constrained}, tapering \cite{guerci1999theory,furrer2006covariance}, {diagonal loading \cite{smith2001adaptive}}, and shrinkage \cite{ledoit2004well,ledoit2012nonlinear,schafer2005shrinkage}) can be much more flexible. However, neither approach easily allows prior knowledge---for instance, known covariance matrices at related conditions---to be incorporated. Estimation in both settings often involves solving an optimization problem, such as maximizing the likelihood of a data set. Defining this objective function requires prescribing a specific probability distribution for the data, which may not be readily available. Moreover, maximum likelihood is not linked to the natural distance on the manifold of symmetric positive matrices. Also, as we shall show later, the resulting estimates can be sensitive to noise.



To overcome these obstacles, we propose to build covariance families by connecting \emph{representative} covariance matrices (called ``anchors'') through geodesics. The resulting covariance classes can thus be tailored to the problem of interest.
Second, as an alternative to selecting the most representative parameter by maximizing the likelihood, we advocate for a differential geometric approach to estimation {within the family}: minimizing the geodesic distance to a sample covariance matrix. Our approach, which we call \emph{natural projection}, is consistent with the notion of distance employed to build the covariance family and does not require assuming a particular probability distribution for the data.
Later, in a case study involving observations of a groundwater flow model, we will show that natural projection can outperform maximum likelihood estimation in the presence of noise.


%


Geodesic interpolation, smoothing, and regression of matrices and related objects have been active topics of research. Several papers \cite{amsallem2008interpolation,amsallem2009method,amsallem2011online} use interpolation on matrix manifolds to adapt and construct reduced-order models,
while others \cite{pennec2006riemannian,lenglet2005riemannian} propose a Riemannian framework for the interpolation and regularization of tensor fields, with broad applications in imaging.
Local polynomial regression in the manifold of symmetric positive-definite matrices has been explored in \cite{yuan2012local}, also in the context of computer vision and medical imaging.
Other authors have pursued higher-order interpolation of positive-definite \cite{absil2016differentiable,gousenbourger2014piecewise} and semi-definite \cite{gousenbourger2017piecewise} matrices using B\'{e}zier curves.
Additionally, \cite{moakher2005differential,moakher2006symmetric} characterize the Riemannian mean of positive-definite matrices and propose a multivariate geodesic interpolation scheme for such matrices using weights.
Our work also relies on geodesic interpolation, but for the purpose of building covariance families.
The idea of differential geometric methods for covariance \emph{estimation} has links to the broad field of information geometry \cite{amari2016information}, which constructs manifolds of probability distributions and analyzes their geometric properties. In this general setting, the likelihood function has been characterized as a notion of distance in \cite{akaike1998information,akaike1971determination}. Matrix nearness using Bregman divergences, and its geometric interpretations, have been discussed thoroughly in \cite{dhillon2007matrix}.

As described above, our covariance families will follow from geodesic interpolation of a given set of anchor matrices. The anchors should be representative of known problem instances---e.g., empirical observations or computational simulations of the relevant spatiotemporal process at related conditions. Combining these instances into a parametric family constitutes a hybrid approach to covariance modeling that can yield much richer and more problem-specific covariances than standard kernels. Using geodesics ensures that the entire covariance family lies in the manifold of symmetric positive-definite matrices.
%
These families can also be interpretable: the internal parameters may serve as explicative variables for the problem of interest.
%
Another advantage of this approach is that it harnesses the asymmetry of information between online and offline stages of a problem: that is, each anchor covariance matrix can be computed offline to a desired accuracy and later used for online estimation with limited data.

Having constructed a geodesically parameterized covariance family, we also analyze different alternatives for estimation within the family---i.e., given a data set, identifying the most ``representative'' member. This problem is usually solved by assigning a probability distribution to the data and selecting the parameter values that maximize the resulting likelihood function. Under some conditions, this process is equivalent to minimizing a particular direction of the Kullback--Leibler (KL) divergence, known as reverse information-projection (reverse I-projection) \cite{csiszar2003information}; this is further equivalent to minimizing Stein's loss \cite{stein1975estimation,stein1986lectures}. Alternatively, one might minimize the opposite direction of KL divergence; this choice is known as I-projection.
Consistent with the construction of the family, we instead propose to use natural projection: selecting the covariance matrix within the family that minimizes the geodesic distance to the sample covariance matrix.
We will show that the other methods are essentially linear approximations of natural projection. In particular, we will show that the estimates produced by natural projection and the two forms of I-projection are locally equivalent up to second order.
%
%
In contrast with other methods, however, the optimality condition for natural projection is equivalent to an orthogonality condition on the tangent space of the matrix manifold.
Since natural projection does not require modeling the distribution of the data, it may be easier to apply in practice. We also show that it can yield reduced sensitivity to noise.
%

{Performing natural projection requires that the sample covariance lie in the manifold of symmetric positive matrices, and thus that the size of the sample $q$ generally be greater than the dimension $n$ of the parameters. Even though the sample covariance is itself a consistent estimator of the population covariance as $q/n \to \infty$, it can be improved upon significantly for finite $q/n$ (say $q/n < O(100)$) \cite{ledoit2004honey,furrer2006covariance,donoho2018optimal,smith2005covariance,marvcenko1967distribution}. This also defines the regime of applicability for many of the estimation (and regularization) methods we propose in this paper. But we also emphasize that the construction of geodesically parameterized covariance families, apart from natural projection, is of independent interest, and that one can perform maximum likelihood estimation within these parametric families for $q < n$. Our goal, also, is not to create a consistent estimators of generic population covariances, but rather to create useful and expressive \emph{parametric models} for covariance matrices and to study estimation \emph{within} these models, in the appropriate sample size regime.}


To summarize, the original contributions of this paper are: (1) to devise a general framework for problem-specific geodesically parameterized covariance families; (2) to propose natural projection as an alternative means of estimation within a covariance family; (3) to analyze the differences between natural projection and other standard estimation techniques; and (4) to demonstrate the advantages of geodesically parameterized families and natural projection in a case study. 


The ability to find the closest member of a covariance family has several further applications. First, consider de-noising or regularization: if a covariance matrix is well approximated by a given family, one can project it to the family to reduce sampling noise. Second, consider efficient storage using a geodesic basis: given a covariance family and a set of related matrices, one could store only the values of the parameters of the closest matrices in the family. As a consequence, storage is reduced only to the optimal parameters and the anchor matrices.



The plan of the paper is as follows. In \Cref{sec:symcone}, we review the geometry of the manifold of symmetric positive-definite matrices, define the notion of a geodesic covariance family, and introduce natural projection alongside some standard alternatives. \Cref{sec:mainres} discusses properties of the optimization problems associated with each of these estimation methods. In \Cref{sec:comparison}, we perform local analyzes that compare natural projection with existing alternatives. \Cref{sec:pvariate} extends the geodesic covariance family construction to general multi-parameter settings. In \Cref{sec:aquifer}, we demonstrate the performance of geodesic covariance families and natural projection in a case study: characterizing the spatial variations of hydraulic head in an aquifer. Conclusions follow in \Cref{sec:conclu}.

\section{Tools for covariance estimation on a geodesic family}\label{sec:symcone}
In \Cref{sec:background}, we recall some results on the geometry of the symmetric positive-definite cone. In \Cref{sec:covfamily}, we introduce the idea of a geodesic covariance family. In \Cref{sec:spectral}, we present the loss functions we will consider for estimation. In \Cref{sec:defii}, we introduce natural projection and contrast it with canonical approaches to parametric covariance estimation.

\subsection{The geometry of the symmetric positive-definite cone}\label{sec:background}
Let $\Sym$ be the space of $n\times n$ symmetric matrices, and let $\Sp$ denote the manifold of symmetric positive-definite $n\times n$ matrices. This manifold has been studied extensively in the literature (e.g., \cite{bhatia2009positive,faraut1994analysis,smith2005covariance,edelman1998geometry}).

Let $X_{A},Y_{A}\in \Sym$ be tangent vectors to $\Sp$ at $A$: $T_{A}\Sp$. The natural metric $g_{A}$ is defined as the inner product in the tangent space at $A$:
\begin{displaymath}
g_{A}(X_{A},Y_{A})=\Tr(X_{A}A^{-1}Y_{A}A^{-1}).
\end{displaymath}
We will denote the tangent vector $X_{A}$ simply as $\underline{X}$ when there is no ambiguity in the choice of tangent space.

The exponential map transports an object in the tangent space to its corresponding element on the manifold, and is defined as:
\begin{displaymath}
B=\exp_{A}(\underline{B})=  A^{\frac{1}{2}}\exp(A^{-\frac{1}{2}}\underline{B}A^{-\frac{1}{2}})A^{\frac{1}{2}},
\end{displaymath}
where $A^{\frac{1}{2}}$ is the symmetric square root.
Conversely, the logarithm map transports objects from the manifold to the tangent space:
\begin{displaymath}
\underline{B}=\log_{A}(B)=  A^{\frac{1}{2}}\log_{m}(A^{-\frac{1}{2}}BA^{-\frac{1}{2}})A^{\frac{1}{2}},\  A^{-\frac{1}{2}}\underline{B}A^{-\frac{1}{2}}= \log_{m}(A^{-\frac{1}{2}}BA^{-\frac{1}{2}}).
\end{displaymath}

Let $A_{1}$ and $A_{2}$ belong to this manifold. Associated with the natural metric, there exists a natural distance $d(A_{1},A_{2})$ that is invariant with respect to matrix inversion:
\begin{equation}\label{eq:inv1}
d(A_{1},A_{2})=d(A_{1}^{-1},A_{2}^{-1}),
\end{equation}
and with respect to congruence via any invertible matrix $Z$:
\begin{equation}\label{eq:inv2}
d(A_{1},A_{2})=d(ZA_{1}Z^{\top},ZA_{2}Z^{\top}).
\end{equation}
Moreover, a parametrization of the geodesic, which at any point minimizes the natural distance to $A_{1}$ and $A_{2}$, is given by:
\begin{equation*}
\varphi_{A_{1} \rightarrow A_{2}}(t)=A_{1}^{\frac{1}{2}}\exp_{m}\bigl(t\log_{m}(A_{1}^{-\frac{1}{2}}A_{2}A_{1}^{-\frac{1}{2}})\bigr)A_{1}^{\frac{1}{2}},
\end{equation*}
where $\varphi_{A_{1} \rightarrow A_{2}}(t)\in \Sp$ for all $t\in \mathbb{R}$. Clearly, $A_{1}^{-\frac{1}{2}}A_{2}A_{1}^{-\frac{1}{2}}$ admits an orthogonal eigendecomposition of the form $U\Lambda U^{\top}$. Notice that $\Lambda$ contains the generalized eigenvalues of the pencil $(A_{2},A_{1})$, which we denote as $\lambda^{(A_2,A_1)}_{k},\;k=1,\dots,n$. Therefore, $\varphi_{A_{1} \rightarrow A_{2}}(t)$ can be expressed as:
\begin{equation}\label{eq:geodesic}
\varphi_{A_{1} \rightarrow A_{2}}(t)=A_{1}^{\frac{1}{2}}\exp_m\bigl(t \log_m(U\Lambda U^{\top})\bigr)A_{1}^{\frac{1}{2}}=A_{1}^{\frac{1}{2}}U\Lambda^tU^{\top}A_{1}^{\frac{1}{2}}.
\end{equation}
Notice that we recover the trivial cases $\varphi_{A_{1} \rightarrow A_{2}}(t=0)=A_{1}$ and $\varphi_{A_{1} \rightarrow A_{2}}(t=1)=A_{2}$; we call $A_1$ and $A_2$ the \emph{anchor} matrices. The geodesic can also be expressed as:
\begin{equation*}
\varphi_{A_{1} \rightarrow A_{2}}(t)=A_1(A_1^{-1}A_2)^{t_1}=A_2(A_2^{-1}A_{1})^{1-t_1}.
\end{equation*}

Additionally, there is a closed-form expression for the natural distance between any two matrices in $\Sp$:
\begin{equation}\label{eq:distance}
d(A_{1},A_{2})=d(A_{1}^{-\frac{1}{2}}A_{2}A_{1}^{-\frac{1}{2}},I)=\|\log_m(A_{1}^{-\frac{1}{2}}A_{2}A_{1}^{-\frac{1}{2}})\|_F=\sqrt{\sum_{k=1}^{n}\log^2\lambda^{(A_2,A_1)}_{k}}.
\end{equation}
From the above, we have:
\begin{equation}\label{eq:proportion}
d\bigl(A_{1},\varphi_{A_{1} \rightarrow A_{2}}(t)\bigr)=\lvert t\rvert d(A_{1},A_{2}).
\end{equation}

\Cref{eq:distance} appears extensively in the literature of differential geometry and inference. Up to a constant, it is known as the Fisher information metric or as Rao's distance \cite{atkinson1981rao,rao1987differential,rao1949appendix}. When measuring distance between covariance matrices, it is also known as the F{\"o}rstner metric \cite{forstner2003metric}.

Unlike the manifold of symmetric positive-definite matrices, the manifold of symmetric positive \emph{semi-definite} matrices does not enjoy as much structure. As shown clearly in \cite{bonnabel2009riemannian} and discussed in \cite{vandereycken2012riemannian}, the main drawback is that there is no notion of distance that satisfies the invariance properties in \Cref{eq:inv1,eq:inv2}.

\subsection{Definition of the geodesic covariance family}\label{sec:covfamily}

As described in \Cref{sec:intro}, the ability to create tailored parametric covariance \emph{functions} out of two (or more) covariance matrices of interest has several potential benefits. First, via the choice of anchor matrices, the covariance function can be made representative of the particular problem at hand; second, the parameters of the covariance function can become meaningful explicatory variables of the spatiotemporal process.

To begin, we define the notion of a one-parameter covariance family. We will generalize this notion to the multi-parameter case in \Cref{sec:pvariate}.

\begin{definition}[\textbf{Covariance function and family}]\label{def:kernel} A $one$-parameter covariance function is a map ${\varphi}:\mathbb{R} \rightarrow \Sp$; its corresponding covariance family is the image of ${\varphi}$.
\end{definition}

The covariance family structure we will employ in \Cref{sec:symcone,sec:mainres,sec:comparison} is a geodesic between anchor matrices.
%
Let $A_1$ and $A_2$ be two elements in $\Sp$. Then $\varphi_{A_1 \rightarrow A_2}(t)$ as defined in \Cref{eq:geodesic} is a one-parameter covariance function, whose image is a covariance family. This covariance function immediately satisfies the following two properties:
\begin{enumerate}
\item $\varphi_{A_{1} \rightarrow A_{2}}^{-1}(t)=\varphi_{A_{1}^{-1} \rightarrow A_{2}^{-1}}(t)$,
\item $\varphi_{A_{1} \rightarrow A_{2}}(t)=\varphi_{A_{2} \rightarrow A_{1}}(1-t)$.
\end{enumerate}


Since the parameter $t \in \mathbb{R}$, the covariance family is an infinitely long curve on the manifold rather than a segment between the two anchor matrices. Notice that the first property allows one to work with precision matrices and still obtain the same results. The second guarantees invariance with respect to the order of the matrices. These properties  make the geodesic a compelling covariance function to be used for practical applications.

%

It is often useful to be able to scale the family to adapt to changes of the magnitude of the problem. Within our covariance family framework, this extra degree of freedom is achieved by building a geodesic between any matrix and the same times a constant.

\begin{remark}[\textbf{Scaling of a covariance function}]\label{lem:scaling}
Let $A_1$ be an element in $\Sp$ and $\alpha \in \mathbb{R}^{+}$. Notice that $\varphi_{A_1 \rightarrow \alpha A_1}(t)$ as defined in \Cref{eq:geodesic} is a one-parameter covariance function of the form $\varphi_{A_1 \rightarrow \alpha A_1}(t)=\alpha^t A_1$. If $A_1$ is replaced by a one-parameter covariance function, the scaling applies to the whole family.
\end{remark}

The scaling factor $\alpha^t$ in \Cref{lem:scaling} is positive for any $t$ in the real line. Clearly, if the scaling is applied to a one-parameter covariance function, we are left with two degrees of freedom: one that moves across the anchor matrices and another that controls the magnitude of the entries.

\subsection{Spectral functions}\label{sec:spectral}
In \Cref{sec:mainres}, we will be interested in selecting the ``most representative'' member of a covariance family given a data set. Doing so will entail minimizing certain loss functions: distances or divergences between distributions. All the loss functions we employ are \emph{spectral functions}.

\begin{definition}[\textbf{Spectral function}]
Let $A_1$ be a matrix in $\Sp$. A function $F(A_1)$ is a \emph{spectral function} if it is a differentiable and symmetric map from the eigenvalues of $A_1$ to the reals. The function $F$ can be understood as a composition of the eigenvalue function $\lambda$ and a differentiable and symmetric map $f$; that is, $F(A_1)=f \circ \lambda (A_1)$.
\end{definition}

Closed-form expressions for some spectral functions that we shall use later are presented below.
\begin{remark}\label{rem:examples1}
The following notions of distance or divergence can be expressed as functions of the generalized eigenvalues $(\lambda_k)_{k=1}^n$ of the pencil $(A_2,A_1)$:
\begin{itemize}
\item Natural distance in $\Sp$:
 \vspace{-0.3em}
\begin{equation}\label{eq:natdist}
d(A_{1},A_{2})=\sqrt{\sum_{k=1}^{n}\log^2\lambda_{k}}.
 \vspace{-0.3em}
\end{equation}
\item Kullback--Leibler divergence between multivariate normals:
\begin{equation}\label{eq:likelihood}
\infdiv[\big]{N(0,A_{1})}{N(0,A_{2})}=\sum_{k=1}^{n}\frac{\lambda_{k}^{-1}+\log\lambda_{k}-1}{2}.
\end{equation}
\item Kullback--Leibler divergence between multivariate normals, swapping the order:
\begin{equation}\label{eq:kl_div}
\infdiv[\big]{N(0,A_{2})}{N(0,A_{1})}=\sum_{k=1}^{n}\frac{\lambda_{k}-\log\lambda_{k}-1}{2}.
\end{equation}
\end{itemize}
\end{remark}


\subsection{Covariance estimation in a geodesic family}\label{sec:defii}
Now we define several alternative optimization problems that each describe  estimation in a geodesic covariance family. We will formally contrast these problems in the next sections. \Cref{fig:testfig2} illustrates the covariance function as a geodesic from $A_1$ and $A_2$ on the manifold of symmetric positive-definite matrices, along with multiple projections (i.e., estimates within the family) of the sample covariance matrix $\widehat{C}$.


\begin{figure}[!ht]
  \centering
  \includegraphics[scale=0.5]{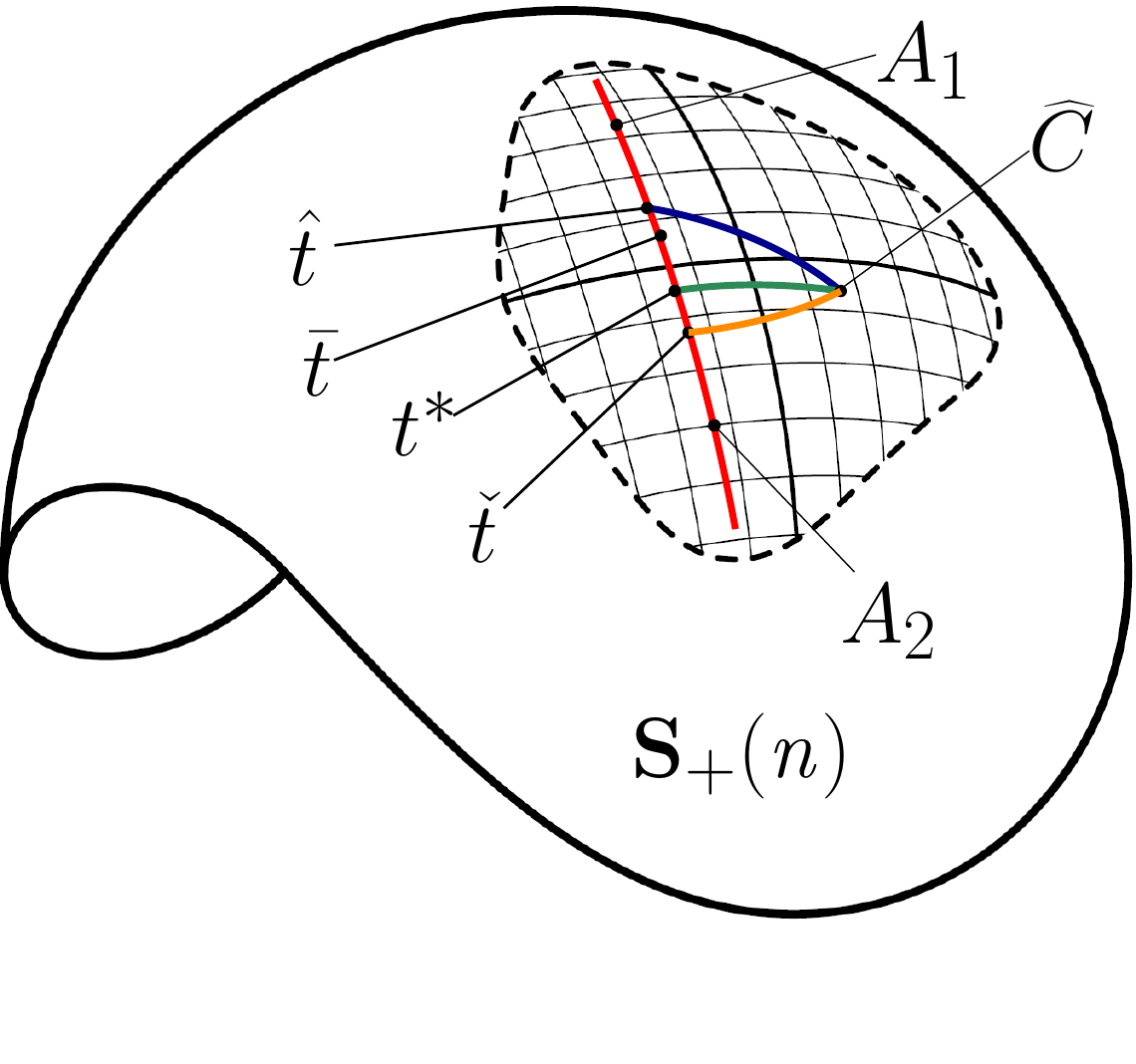}
  \caption{Representation of the geodesic between anchors $A_1$ and $A_2$ and projection of the sample covariance matrix $\widehat{C}$ as solution of \Cref{prob:max,prob:reviproj,prob:iproj,prob:min}. We denote $\hat{t}$ as the value of $t$ that maximizes the likelihood (reverse I-projection), $\check{t}$ the I-projection, and $t^{*}$ the natural projection (cf.~\Cref{lem:idem}).}
  \label{fig:testfig2}
\end{figure}

Let $y_1,\dots, y_q$ be independent and identically distributed observations from some distribution with density function $p_Y\bigl(\cdot \, ; \varphi_{A_{1} \rightarrow A_{2}}(t)\bigr)$, parameterized by $\varphi$. In general, the maximum likelihood estimate of $t$ is then
\begin{align}
t^{\text{ML}} \in \argmax_{t \in (-\infty,\infty)}\;& \sum_{i=1}^{q} \log p_{Y}(y_i; \varphi_{A_{1} \rightarrow A_{2}}(t) ).
\label{eq:mlegeneral}
\end{align}
For comparison with other techniques, we consider the specific case where $p_Y$ is multivariate Gaussian and, without loss of generality, zero mean. In this case, the maximum likelihood estimate is as follows:


\begin{problem}[\textbf{Maximum likelihood with a Gaussian model}]\label{prob:max}
\begin{align*}
y_i \stackrel{iid}{\sim} N&\bigl(0,\varphi_{A_{1} \rightarrow A_{2}}(t)\bigr) \\
\hat{t}\in \argmax_{t \in (-\infty,\infty)} \;& \sum_{i=1}^{q} \log N\bigl(y_i; 0, \varphi_{A_{1} \rightarrow A_{2}}(t)\bigr). \nonumber
\end{align*}
\end{problem}

{Note that \eqref{eq:mlegeneral} and \Cref{prob:max} can in general be solved for $q < n$. In other words, maximum likelihood estimation within the geodesic covariance family, under some distributional assumption on the data, does not require a full rank sample covariance $\widehat{C} \coloneqq \frac{1}{q-1} \sum_{i=1}^q (y_i - \bar{y}) (y_i - \bar{y})^\top $.}

Now, suppose that the sample covariance matrix $\widehat{C}$ of zero-mean Gaussian data $(y_i)_{i=1}^q$ is full rank, which is typically the case when $q\ge n$. In this case, the solution to \Cref{prob:max} is equivalent to that obtained by reverse I-projection, which consists of minimizing the KL divergence as follows. This equivalence is shown in \Cref{lem:convergence}.

\begin{problem}[\textbf{Reverse I-projection}]\label{prob:reviproj}
\begin{align}
\hat{t}\in \argmin_{t\in (-\infty,\infty)}\;\infdiv[\Big]{N(0,\widehat{C})}{N\bigl(0,\varphi_{A_{1} \rightarrow A_{2}}(t)\bigr)}.
\end{align}
\end{problem}
Since \Cref{prob:max,prob:reviproj} are equivalent, we will only refer to \Cref{prob:reviproj} going forward.
We will also consider \emph{I-projection}, which minimizes the KL divergence but in the opposite order:

\begin{problem}[\textbf{I-projection}]\label{prob:iproj}
\begin{displaymath}
\check{t}\in \argmin_{t\in (-\infty,\infty)}\;\infdiv[\Big]{N\bigl(0,\varphi_{A_{1} \rightarrow A_{2}}(t)\bigr)}{N(0,\widehat{C})}.
\end{displaymath}
\end{problem}
Note that the KL divergence is proportional to Stein's loss \cite{ledoit2018optimal}. Therefore, \Cref{prob:reviproj,prob:iproj} can also be cast as minimizing Stein's loss with the appropriate order of the arguments.

Finally, we explore the possibility of covariance estimation by minimizing the geodesic distance $d(\cdot, \cdot)$ defined in \Cref{eq:distance}:
\begin{problem}[\textbf{Natural projection}]\label{prob:min}
\begin{displaymath}
{t}^{*}\in \argmin_{t\in (-\infty,\infty)}\; d\bigl(\varphi_{A_{1} \rightarrow A_{2}}(t),\widehat{C}\bigr).
\end{displaymath}
\end{problem}

\subsection{{Choosing the anchor matrices}}
Before proceeding with further analysis, we offer a few comments on the choice of the anchor covariance matrices $A_1$ and $A_2$ defining the parametric covariance family. (These comments are equally applicable to the multi-parameter case of \Cref{sec:pvariate}, with more than two anchor matrices.) In general, we view the choice of anchors as a \emph{modeling} issue, tied to the purpose of the family and the availability of relevant information. Yet it can be shaped by the following principles.
Clearly, anchors should be chosen such that they are representative of the spatiotemporal processes that the covariance family will be used to describe.
How to do this? In general, we suggest that the anchors should be chosen and connected in a way that reflects any \emph{possible} latent or explanatory variables that describe the problem. For instance, in the example of \Cref{sec:aquifer}, we choose the anchors to span a range of parameters describing statistics of the input to a stochastic PDE. The resulting geodesic parameters (connecting covariances of the PDE \emph{solution} field) do not necessarily correspond directly to the input's parameters, but they do encompass a continuum of values and thus relevant behaviors.

{While in some problems, the explanatory variables will be clear, there are  other problems where some explanatory variables may not be apparent. In this case, a strategy would be to collect anchors that capture all ``qualitatively different'' patterns of covariance. The resulting covariance family will, by construction, be able to interpolate continuously between these different regimes. For example, in mathematical finance, modeling the covariance of multiple asset prices is essential to mitigating volatility through diversification \cite{markowitz1952portfolio}. To construct a useful model, one can connect multiple covariance matrices corresponding to different market conditions (e.g., time of the year, prevailing Federal Reserve interest rates) through geodesics to create a richer family of covariances suitable for any condition.}

\section{Properties of the covariance estimation problem, one-parameter case}\label{sec:mainres}

In this section, we characterize various properties of \Cref{prob:reviproj,prob:iproj,prob:min}. Our main results are the optimality conditions and their corresponding geometric interpretations. These results are presented in \Cref{thm:solgeom,thm:sollike,thm:soliproj}, proofs of which are deferred to \Cref{sec:technical}. First, however, we establish uniqueness of the optima and idempotence of the associated projections. Supporting results (\Cref{lem:spectral,lem:convergence,lem:convexity,lem:convexconcave}) and their proofs are also deferred to \Cref{sec:technical}.


\begin{lemma}[\textbf{Uniqueness of the solution}]\label{lem:unique}
Each of \Cref{prob:reviproj,prob:iproj,prob:min} has a unique solution.
\end{lemma}
\begin{proof}
Since the optimization problems are unconstrained, uniqueness follows immediately from the convexity of \cref{prob:min} (shown in \Cref{lem:convexity}) and of \Cref{prob:reviproj,prob:iproj} (shown in \Cref{lem:convexconcave}).
\end{proof}

Since the solutions are unique, though in general distinct (see below), we will use $\hat{t}$ to denote the result of reverse I-projection, $\check{t}$ that of  I-projection, and $t^{*}$ that of natural projection. Uniqueness also allows defining the distance between the sample covariance matrix and the geodesic as $d\bigl(\varphi_{A_{1} \rightarrow A_{2}}(t^{*}),\widehat{C}\bigr)$.

If the sample covariance matrix already belongs to the covariance family, the most representative member of the family ought to be the sample covariance matrix itself. This result holds true for all three problems.

\begin{lemma}[\textbf{Idempotence of projections}]\label{lem:idem}
If $\widehat{C}\in \varphi_{A_{1} \rightarrow A_{2}}(t)$, then there exists a unique $\bar{t}$ such that $(\lambda_k^{(A_{2},A_{1})})^{\bar{t}}=\lambda_k^{(\widehat{C},A_{1})}$,  $k=1,\dots n$, where $\lambda_k^{(A_{2},A_{1})}$ and $\lambda_k^{(\widehat{C},A_{1})}$ are the $k$-th eigenvalues of the pencil $(A_{2},A_{1})$ and $(\widehat{C},A_1)$, respectively. Moreover, under this condition, $\bar{t}=t^{*}=\hat{t}=\check{t}$ (cf.~\Cref{fig:testfig2}), $\widehat{C}=\varphi_{A_{1} \rightarrow A_{2}}(\bar{t})$ with:
  \begin{displaymath}
    \bar{t}=\frac{\sum_{k=1}^{n} \log\lambda_k^{\widehat{C}}-\sum_{k=1}^{n} \log\lambda_k^{A_{1}}}{\sum_{k=1}^{n} \log\lambda_k^{A_{2}}-\sum_{k=1}^{n} \log\lambda_k^{A_{1}}},
  \end{displaymath}
 where $\lambda_k^{A_{1}}$, $\lambda_k^{A_{2}}$, and $\lambda_k^{\widehat{C}}$, are the $k$-th eigenvalues of $A_{1}$, $A_{2}$, and $\widehat{C}$, respectively. This expression also holds when $A_{1}=\alpha A_{2}$, for any $\widehat{C}$ and $\alpha>0$.
\end{lemma}
\begin{proof}
If $\widehat{C}\in \varphi_{A_{1} \rightarrow A_{2}}(t)$, then there exists a $\bar{t}$ such that $\widehat{C}=A_{1}^{\frac{1}{2}}U\Lambda^{\bar{t}}U^{\top}A_{1}^{\frac{1}{2}}$. Rearranging the terms, we obtain $A_{1}^{-\frac{1}{2}}\widehat{C}A_{1}^{-\frac{1}{2}}=U\Lambda^{\bar{t}}U^{\top}$, which is an eigendecomposition of $A_{1}^{-\frac{1}{2}}\widehat{C}A_{1}^{-\frac{1}{2}}$. This is equivalent to say that $\Lambda^{\bar{t}}$ contains the $\lambda_k^{(\widehat{C},A_1)}$. But also, from our notation in \Cref{eq:geodesic} we knew that $\Lambda$ contains the $\lambda_k^{(A_{2},A_{1})}$.

Notice that $\bar{t}$ satisfies the general form (cf.~\Cref{eq:soldistance}) of distance minimization, and similarly for the reverse I-projection (cf.~\Cref{eq:solreviproj}) and the I-projection (cf.~\Cref{eq:soliproj}). Using uniqueness in \Cref{lem:unique}, we conclude that $\bar{t}=t^{*}=\hat{t}=\check{t}$. The last $\widehat{C}=\varphi_{A_{1} \rightarrow A_{2}}(\bar{t})$ follows by the definition of $\bar{t}$. For the closed form solution, set $\widehat{C}=\varphi_{A_{1} \rightarrow A_{2}}(t^*)=A_{1}^{\frac{1}{2}}U\Lambda^{t^*}U^{\top}A_{1}^{\frac{1}{2}}=A_{1}^{\frac{1}{2}} (A_{1}^{-\frac{1}{2}}A_{2}A_{1}^{-\frac{1}{2}})^{t^*}A^{\frac{1}{2}}$. Now, take the determinant of both sides and apply its properties to have $\det{(\widehat{C})}=\det{(A_{1})}\det(A_{1}^{-\frac{1}{2}}A_{2}A_{1}^{-\frac{1}{2}})^{t^{*}}$. Applying logarithms on both sides we obtain the desired result. Clearly, it solves \cref{prob:min} since distance in this case is zero. Finally, if $A_{1}=\alpha A_{2}$, then $\Lambda=\alpha I$ and note that the general expression simplifies to $\Tr\bigl(\log_m(\Lambda^{t^*}\widehat{C}^{-\frac{1}{2}}A_{1}\widehat{C}^{-\frac{1}{2}})\bigr)=0$. Using Jacobi's formula, it simplifies further to $\det{(\Lambda^{t^*}\widehat{C}^{-\frac{1}{2}}A_{1}\widehat{C}^{-\frac{1}{2}})}=1$ and we come back to the same above expression with determinants.
\end{proof}

The following propositions characterize the optimal solutions of \Cref{prob:reviproj,prob:iproj,prob:min}.

\begin{proposition}[\textbf{Natural projection for covariance estimation}]\label{thm:solgeom}
The optimal parameter $t^*$ satisfies:
  \begin{equation}\label{eq:minopt}
    \Tr\bigl(\log_m(Z\Lambda^{-t^*})\log_m(\Lambda)\bigr)=0,
  \end{equation}
  where $Z=U^{\top}A_1^{-\frac{1}{2}}\widehat{C}A_1^{-\frac{1}{2}}U$.
  This optimality equation can also be rewritten as an orthogonality condition on the tangent space:
\begin{equation}\label{eq:ortho}
  g_{R_{A_{1} \rightarrow A_{2}}(t^*)}\bigl(\underline{\strut{A_1^{-\frac{1}{2}}\widehat{C}A_1^{-\frac{1}{2}}}},\underline{\strut{R_{A_{1} \rightarrow A_{2}}(1+t^*)}}\bigr)=0,
\end{equation}
where $R_{A_{1} \rightarrow A_{2}}(t)=(A_1^{-\frac{1}{2}}A_2A_1^{-\frac{1}{2}})^t$ is the whitened geodesic $A_1^{-\frac{1}{2}}\varphi_{A_{1} \rightarrow A_{2}}A_1^{-\frac{1}{2}}$.
\end{proposition}

The natural projection consists in minimizing the natural distance to a certain matrix over a curve. This is similar to what one would do in an Euclidean space, but on a manifold. On the tangent space, this operation looks like finding a point at which the projected geodesic is orthogonal to the direction of the outside matrix (\Cref{eq:ortho}). \Cref{fig:tangent} illustrates this relationship.

\begin{figure}[!ht]
  \centering
  \includegraphics[scale=0.65]{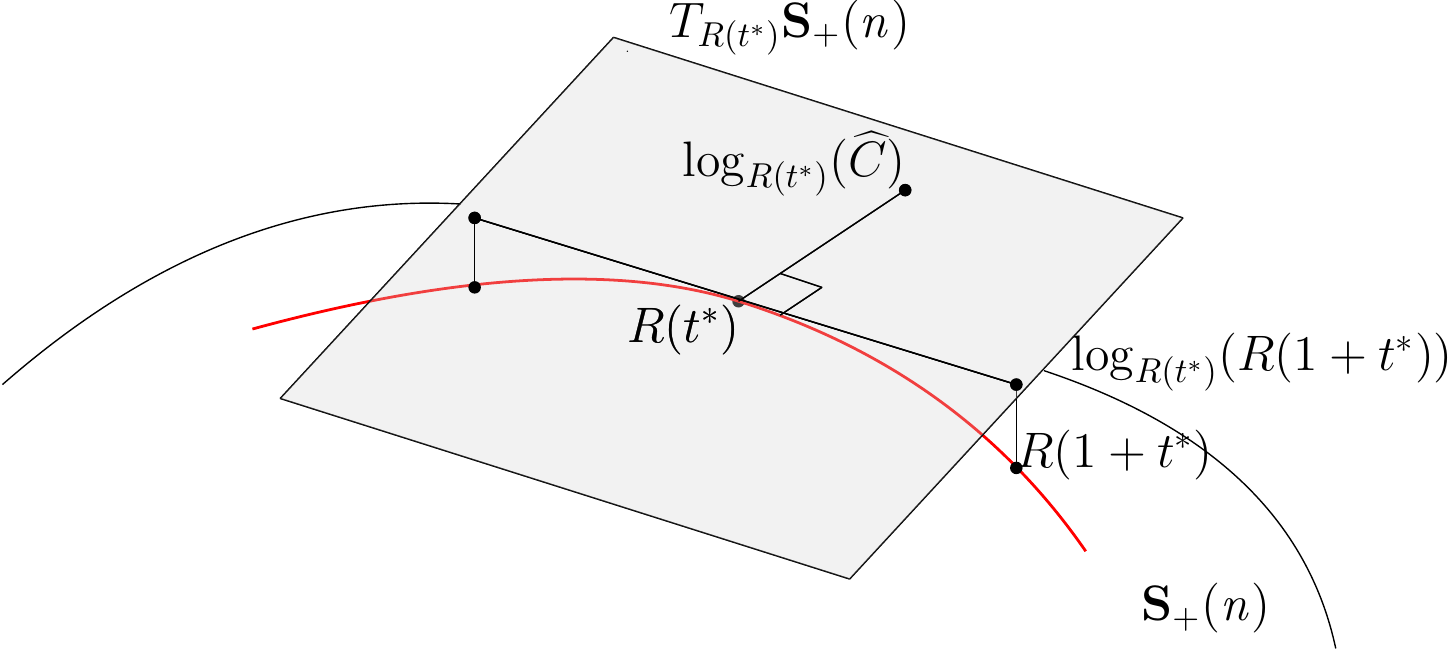}
  \caption{Illustration of the orthogonality condition in \Cref{thm:solgeom}. Solving \cref{prob:min} is equivalent to finding a $t^*$ such that $\widehat{C}$ is orthogonal to the unitary vector $R(1+t^{*})$.}
  \label{fig:tangent}
\end{figure}

As mentioned in \Cref{sec:symcone}, the Fisher information metric between two normal distributions with known mean is proportional to the natural distance, so the former is also minimized when the latter is. Therefore, the aforementioned $t^*$ minimizes the Fisher information metric between $N\bigl(0,\varphi_{A_{1} \rightarrow A_{2}}(t)\bigr)$ and $N(0,\widehat{C})$.

\begin{proposition}[\textbf{Reverse I-projection for covariance estimation}]\label{thm:sollike}
The optimal parameter $\hat{t}$ satisfies:
  \begin{equation}\label{eq:maxopt}
    \Tr\big((Z\Lambda^{-\hat{t}}-\Id)\log_m(\Lambda)\big)=0,
  \end{equation}
  where $Z=U^{\top}A_1^{-\frac{1}{2}}\widehat{C}A_1^{-\frac{1}{2}}U$. This optimality equation can also be rewritten as an orthogonality condition:
\begin{equation}\label{eq:ortho2}
g_{R_{A_{1} \rightarrow A_{2}}(\hat{t})}\Bigl(\underline{{\exp_{R_{A_{1} \rightarrow A_{2}}(\hat{t})}\bigl(A_1^{-\frac{1}{2}}\widehat{C}A_1^{-\frac{1}{2}}-R_{A_{1} \rightarrow A_{2}}(\hat{t})\bigr)}},\underline{
\mathop{\vphantom{{\exp_{R_{A_{1} \rightarrow A_{2}}(\hat{t})}\bigl(A_1^{-\frac{1}{2}}\widehat{C}A_1^{-\frac{1}{2}}-R_{A_{1} \rightarrow A_{2}}(\hat{t})\bigr)}}}
{R_{A_{1} \rightarrow A_{2}}(1+\hat{t})}}\Bigr)=0.
\end{equation}
\end{proposition}

Notice that the matrix $A_1^{-\frac{1}{2}}\widehat{C}A_1^{-\frac{1}{2}}-R_{A_{1} \rightarrow A_{2}}(\hat{t})$ lives in a Euclidean space and not necessarily in $\Sp$. Indeed, such a subtraction is the usual way of obtaining a vector between two points in a ``flat'' space, but it does not necessarily preserve positive-definiteness. Intuitively, the likelihood (which yields the first argument of \Cref{eq:ortho2}) seems to be a flat notion, whereas the geodesic (yielding the second argument) is in the manifold; thus one could argue that reverse I-projection is actually inconsistent. Instead, one should either maximize the likelihood over a family produced by a convex combination of anchors (i.e., both the family and the divergence we are minimizing are in a flat space) or minimize the natural distance over a proper geodesic (as in \Cref{prob:min}). Indeed, the orthogonality condition for natural projection in \Cref{thm:solgeom} is far more direct.

We can develop similar results for \Cref{prob:iproj}:
\begin{proposition}[\textbf{I-projection for covariance estimation}]\label{thm:soliproj}
The optimal parameter $\check{t}$ satisfies:
  \begin{equation}\label{eq:miniproj}
    \Tr\bigl((\Lambda^{\check{t}}Z^{-1}-\Id)\log_m(\Lambda)\bigr)=0,
  \end{equation}
  where $Z=U^{\top}A_1^{-\frac{1}{2}}\widehat{C}A_1^{-\frac{1}{2}}U$. This optimality equation can also be rewritten as an orthogonality condition:
\begin{equation}\label{eq:ortho3}
g_{R_{A_{1} \rightarrow A_{2}}(-\check{t})}\Bigl(\underline{\exp_{R_{A_{1} \rightarrow A_{2}}(-\check{t})}\bigl({A_1^{\frac{1}{2}}\widehat{C}^{-1}A_1^{\frac{1}{2}}-R_{A_{1} \rightarrow A_{2}}(-\check{t})}\bigr)},\underline{
\mathop{\vphantom{\exp_{R_{A_{1} \rightarrow A_{2}}(-\check{t})}\bigl({A_1^{\frac{1}{2}}\widehat{C}^{-1}A_1^{\frac{1}{2}}-R_{A_{1} \rightarrow A_{2}}(-\check{t})}\bigr)}}
R_{A_{1} \rightarrow A_{2}}(1-\check{t})}\Bigr)=0.
\end{equation}
\end{proposition}

Notice that \Cref{eq:maxopt} is the first order Taylor expansion of \Cref{eq:minopt} around the identity matrix if we use $Z\Lambda^{-{t}}$
as a variable. In this sense, maximizing the likelihood (reverse I-projection) corresponds to solving a linearized version of the natural distance minimization problem. The same can be observed with the I-projection if we Taylor expand in $\Lambda^{t}Z^{-1}$ (cf.~\Cref{eq:miniproj,eq:minopt}). It suffices to adapt \Cref{eq:minopt} using $\log_m(Z\Lambda^{-t})=-\log_m(\Lambda^{t}Z^{-1})$.

\section{Local analysis and comparison of the projections}\label{sec:comparison}

We now compare the solutions of \Cref{prob:reviproj,prob:iproj,prob:min} when $\widehat{C}$ is very close to the geodesic (in terms of natural distance).

\begin{lemma}[\textbf{Equality of the limit}]\label{lem:limit}
Let $A_1$, $A_2$, and $C$ be matrices in $\Sp$. Assume that $d\big(\varphi_{A_{1} \rightarrow A_{2}}(t),C\big)=\epsilon$, $\epsilon>0$. In the limit of $\epsilon \to 0$, \Cref{prob:reviproj,prob:iproj,prob:min} have the same solution.
\end{lemma}
\begin{proof}
Let $t^{*}$ be the minimizer of $d\big(\varphi_{A_{1} \rightarrow A_{2}}(t),C\big)$ and $A^{*}=\varphi_{A_{1} \rightarrow A_{2}}(t^{*})$. Without loss of generality, define $\widehat{C}$ such that $d(A^{*},\widehat{C})=1$ and $\varphi_{A^{*} \rightarrow \widehat{C}}(\epsilon)=C$. By the properties of the natural distance, we know that $d(A^{*},{C})=\epsilon$.
Define $Z_{\epsilon}=U^{\top} A_1^{-\frac{1}{2}} \varphi_{A^{*} \rightarrow \widehat{C}}(\epsilon) A_1^{-\frac{1}{2}}  U$ and notice that:
\begin{displaymath}
Z=\lim_{\epsilon \to 0} Z_{\epsilon}=\lim_{\epsilon \to 0} U^{\top} A_1^{-\frac{1}{2}} \varphi_{A^{*} \rightarrow \widehat{C}}(\epsilon) A_1^{-\frac{1}{2}}  U=U^{\top} A_1^{-\frac{1}{2}} A^{*} A_1^{-\frac{1}{2}}  U. \end{displaymath}

By definition, $A^{*}=\varphi_{A_{1} \rightarrow A_{2}}(t^{*})=A_1^{\frac{1}{2}}U\Lambda^{t^{*}}U^{\top}A_1^{\frac{1}{2}}$, thus $Z=\Lambda^{t^{*}}$. The proof is concluded after realizing that, with this value of $Z$, \Cref{eq:maxopt,eq:miniproj} hold.
\end{proof}

Together with idempotence of the projection (\Cref{lem:idem}), \Cref{lem:limit} implies continuity of $\hat{\Delta} t \coloneqq \hat{t}-t^{*}$ at $\epsilon=0$. Indeed, idempotence means pointwise equivalence at $\epsilon=0$, and at the limit $\epsilon \to 0$, we also see $\hat{\Delta} t = 0$. Thus $\hat{\Delta} t$ is continuous at that point. The same is also true for $\check{\Delta} t \coloneqq \check{t}-t^{*}$.

\begin{theorem}[\textbf{Natural projection versus I-projection and reverse I-projection}]\label{thm:comp}
Let $A_1$, $A_2$, $C\in \Sp$, and without loss of generality suppose that $A_1$ is the matrix in $\varphi_{A_{1}\rightarrow A_{2}}$ that minimizes the distance $d\big(\varphi_{A_{1}\rightarrow A_{2}}(t),C\big)$.
The difference between the solutions of \Cref{prob:reviproj,prob:min} as a function of $\epsilon=d\big(\varphi_{A_{1} \rightarrow A_{2}}(t),C\big)$ is $\hat{\Delta} t(\epsilon)$, defined implicitly as:
  \begin{equation}\label{eq:comp}
    \Tr\big((U^{\top}V\Sigma^{\epsilon}V^{\top}U\Lambda^{\hat{\Delta}{t}(\epsilon)}-\Id)\log_m(\Lambda)\big)=0,
  \end{equation}
 where $V\Sigma V^{\top}=A_1^{-\frac{1}{2}}CA_1^{-\frac{1}{2}}$ and $U\Lambda U^{\top}=A_1^{-\frac{1}{2}}A_2A_1^{-\frac{1}{2}}$ are orthogonal eigendecompositions.

Similarly, the difference in the solutions of \Cref{prob:iproj,prob:min} as a function of $\epsilon$ is $\check{\Delta} t(\epsilon)$, defined implicitly as:
  \begin{equation}\label{eq:comp2}
    \Tr\big((\Lambda^{-\check{\Delta}{t}(\epsilon)}U^{\top}V\Sigma^{-\epsilon}V^{\top}U-\Id)\log_m(\Lambda)\big)=0.
  \end{equation}

Moreover, the functions $\hat{\Delta} t(\epsilon)$ and $\check{\Delta} t(\epsilon)$ are continuous at $\epsilon=0$, and
\begin{equation}\label{eq:product}
\hat{\Delta}t'(0)=\check{\Delta}t'(0)=-\frac{g_{A_1}(\underline{ \mathop{\vphantom{A_{2}}}C},\underline{A_{2}})}{d(A_1,A_2)}=0.
\end{equation}
\end{theorem}
\begin{proof}
Refer to the construction of the proof in \Cref{lem:limit} and notice that $A^{*}=A_1$, $t^{*}=0$ for any $\epsilon$, and,
\begin{displaymath}
Z_{\epsilon}=U^{\top} A_1^{-\frac{1}{2}} \varphi_{A^{*} \rightarrow \widehat{C}}(\epsilon) A_1^{-\frac{1}{2}}  U=U^{\top} V\Sigma^{\epsilon}V^{\top}  U.
\end{displaymath}
Then \Cref{eq:comp} follows immediately from \Cref{eq:maxopt}, and \Cref{eq:comp2} from \Cref{eq:miniproj}.
Continuity at $\epsilon=0$ follows from \Cref{lem:idem,lem:limit}.
Now, we can take derivatives of \Cref{eq:comp} and obtain the following:
\begin{equation}\label{eq:compprime}
\Tr\Big(\big(U^{\top}V\Sigma^{\epsilon}\log_m(\Sigma)V^{\top}U\Lambda^{\hat{\Delta}{t}(\epsilon)}+U^{\top}V\Sigma^{\epsilon}V^{\top}U\log_m(\Lambda)\Lambda^{\hat{\Delta}{t}(\epsilon)}\hat{\Delta}{t}'(\epsilon)\big)\log_m(\Lambda)\Big)=0,
\end{equation}
which evaluated at $\epsilon=0$ results in:
\begin{equation}\label{eq:compprimezero}
\hat{\Delta}{t}'(0)=-\frac{\Tr\Big(\big(U^{\top}V\log_m(\Sigma)V^{\top}U\big)\log_m(\Lambda)\Big)}{\Tr\big(\log_m^2(\Lambda)\big)},
\end{equation}
which can be rewritten as:
\begin{displaymath}
\hat{\Delta}{t}'(0)=-\frac{\Tr\big(\log_m(A_1^{-\frac{1}{2}}CA_1^{-\frac{1}{2}})\log_m(A_1^{-\frac{1}{2}}A_2A_1^{-\frac{1}{2}})\big)}{d(A_1,A_2)}=0.
\end{displaymath}
Notice that $\hat{\Delta}{t}'(0)$ vanishes since the numerator is \Cref{eq:ortho}. An analogous derivation for $\check{\Delta}{t}'(0)$ provides the same result.
\end{proof}

From \Cref{eq:product}, note that $\hat{\Delta}{t}'(0)$ can be understood as the inner product of the tangent vectors at $A_1$ pointing to $C$ and $A_2$, normalized by the distance from the reference point to the latter matrix. The expression is analogous to the classic form of the inner product as a product of modulus and angle. In our setting, the modulus is the $d(A_1,A_2)$ and the angle is $\hat{\Delta}{t}'(0)$.

Since $\hat{\Delta}{t}(0)=0$ and $\hat{\Delta}{t}'(0)=0$, a second order Taylor series expansion around $\epsilon=0$ would be:
\begin{displaymath}
\hat{\Delta}{t}(\epsilon)=\frac{\hat{\Delta}{t}''(0)}{2}\epsilon^2+\mathcal{O}(\epsilon^3),
\end{displaymath}
and the same expansion applies for $\check{\Delta}{t}$.

Finally, we compare I-projection and reverse I-projection, summarizing the results below:

\begin{theorem}[\textbf{I-projection versus reverse I-projection}]\label{thm:comp2}
Refer to the notation in \Cref{thm:comp}. The $i$-th derivatives of $\hat{\Delta}{t}$ and $\check{\Delta}{t}$ satisfy the following:
\begin{eqnarray*}
\hat{\Delta}{t}^{(i)}(0) & = & \check{\Delta}{t}^{(i)}(0),\; i=1,3,5\dots \\
\hat{\Delta}{t}^{(i)}(0) & = & -\check{\Delta}{t}^{(i)}(0),\;  i=0,2,4,6\dots .
\end{eqnarray*}
Thus, the Taylor expansion of the difference in the solutions of \Cref{prob:reviproj,prob:iproj} as a function of $\epsilon=d(\varphi_{A_{1} \rightarrow A_{2}}(t),C)$, in a neighborhood of $\epsilon=0$, always attains one additional order of accuracy. In particular,
\begin{displaymath}
\hat{t}(\epsilon)-\check{t}(\epsilon)=\hat{\Delta}{t}''(0)\epsilon^{2}+\mathcal{O}(\epsilon^4),
\end{displaymath}
where
\begin{displaymath}\label{eq:second}
\hat{\Delta}{t}''(0)=\frac{\Tr\big(\log_m^2(A_1^{-\frac{1}{2}}CA_1^{-\frac{1}{2}})\log_m(A_1^{-\frac{1}{2}}A_2A_1^{-\frac{1}{2}})\big)}{d(A_1,A_2)}.
\end{displaymath}
\end{theorem}
\begin{proof}
Comparing the implicit derivative of $\check{\Delta}{t}$ with \Cref{eq:compprime}, we notice that the signs will alternate in each subsequent derivative.
\Cref{eq:second} is obtained after taking derivatives of \Cref{eq:compprime}.
\end{proof}

From \Cref{thm:comp,thm:comp2}, the difference between the solution of \Cref{prob:reviproj} or \Cref{prob:iproj} and that of \Cref{prob:min} is locally of order $\epsilon^2$. Similarly, the difference between the solutions of \Cref{prob:reviproj} and \Cref{prob:iproj} is also of order $\epsilon^2$. Moreover, as shown in \Cref{fig:localanalysis}, given that the first derivative is zero and $\hat{\Delta}{t}^{''}(0)=-\check{\Delta}{t}^{''}(0)$ (\Cref{thm:comp2}), the natural projection will typically fall \textit{between} the I-projection and the reverse I-projection.

\begin{figure}[!ht]
  \centering
  \label{fig:localanalysis}\includegraphics[scale=0.7]{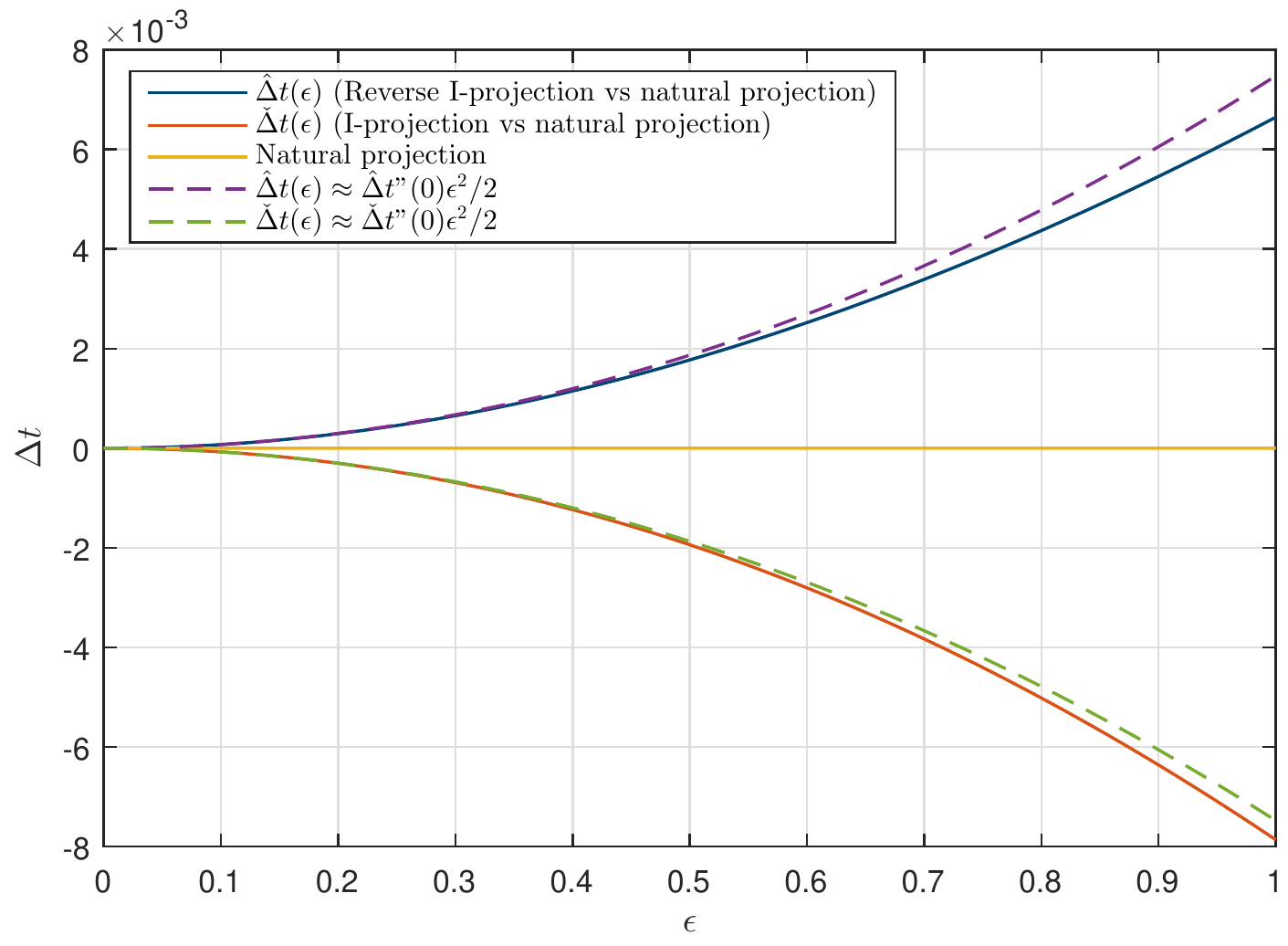}
  \caption{We draw three matrix realizations from a Wishart distribution of size $n=10$ ($A$, $A_2$, and $C$). $A_1$ is then constructed as the minimizer of the natural distance from $C$
  to the geodesic $\varphi_{A \rightarrow A_{2}}(t)$
  (cf.~construction used in \Cref{thm:comp}). We show $\hat{\Delta}{t}$ and $\check{\Delta}{t}$ as a function of the distance $\epsilon$ from $A_1$ to $C$. We control $\epsilon$ by defining $\widehat{C}=\varphi_{A_1 \rightarrow C}\big{(}\frac{1}{d(A_1,C)}\big{)}$; thus $d(A_1,\widehat{C})=1$ and $\epsilon=d(A_1,\varphi_{A_1 \rightarrow \widehat{C}}(\epsilon))$ $($cf.~\Cref{eq:proportion}$)$.
  By construction, $\Delta{t}$ for the natural projection is always zero.}
\end{figure}

Besides the fact that it is a ``middle point'' between I-projection and reverse I-projection, there are other reasons to prefer natural projection. First, as noted earlier, it does not require assigning a probability distribution to the data. Second, the natural projection inherits the invariance properties of the natural distance, the most important of which are symmetry of the arguments and invariance to inversion; the latter property is particularly useful when working with precision matrices, e.g., to take advantage of sparsity. (Reverse) I-projection does not enjoy these properties. Third, as we showed in \Cref{sec:mainres}, the optimality conditions of the I-projections are first order Taylor expansions of the natural projection; therefore, by maximizing the likelihood we are only solving a ``flat'' version of the geodesic problem.
Finally, the natural distance is equivalent (up to a constant) to the Fisher information metric/Rao distance between two normal distributions with common mean, and thus the natural projection also minimizes these loss functions.
If one uses geodesics (lines that minimize the natural distance) to build a covariance family, it is consistent to use the natural projection to select the most representative member of the family. In \Cref{sec:aquifer}, we will show that these advantages of natural projection translate into modeling benefits in practical applications, e.g., robustness to noise-corrupted data.

{Up to now we have not been concerned with asymptotics in the sample size $q$, but it is worth recalling that as ${q}/{n} \to \infty$, the sample covariance matrix converges to the true (population) covariance. If the former is close to the family, we have seen that minimizing the natural distance, maximizing the likelihood, and performing I-projection coincide up to second order. But it is additionally true that, if the true covariance matrix is a member of the geodesic covariance family, natural projection of the sample covariance yields a consistent estimator of the population covariance, i.e., an estimate that converges in probability to the true covariance as $q \to \infty$. For a precise statement and proof of this result, see \Cref{lem:consistency}.}

\section{Generalization to $p$-parameter covariance families}\label{sec:pvariate}

Thus far, we have only presented results for one-parameter covariance functions and families. In this section, we present a generalization to the multi-parameter case. We employ geodesics to construct a function of $p$ parameters using $p+1$ matrices in $\Sp$.
%

\begin{definition}[\textbf{The \emph{unbalanced} $p$-parameter covariance family}]\label{def:pkernel}
By combining two one-parameter covariance functions, we obtain:
\begin{displaymath}
\varphi_{A_{1} \rightarrow A_{2}\rightarrow A_{3}}(t_1,t_2)\coloneqq\varphi_{\big(\varphi_{A_{1} \rightarrow A_{2}}(t_1)\big)\rightarrow A_{3}}(t_2).
\end{displaymath}
Analogously, by combining three one-parameter covariance functions, we obtain:
\begin{displaymath}
\varphi_{A_{1} \rightarrow A_{2}\rightarrow A_{3}\rightarrow A_{4}}(t_1,t_2,t_3)\coloneqq\varphi_{\big(\varphi_{\bigl(\varphi_{A_{1} \rightarrow A_{2}}(t_1)\bigr)\rightarrow A_{3}}(t_2)\big)\rightarrow A_{4}}(t_3).
\end{displaymath}
Recursively, we can construct the unbalanced $p$-parameter covariance function, which we denote as:
\begin{displaymath}
\varphi_{A_{1} \rightarrow \dots\rightarrow A_{p+1}}(t_1,\dots,t_p).
\end{displaymath}
The image of the resulting function is the \emph{unbalanced} $p$-parameter covariance family.
\end{definition}

\begin{definition}[\textbf{The \emph{balanced} $p$-parameter covariance family}]\label{def:pkernelbalanced}
By combining two one-parameter covariance functions, we obtain:
\begin{displaymath}
\varphi_{(A_{1} \rightarrow A_{2})\rightarrow (A_{3}\rightarrow A_{4})}(t_1,t_2,t_3)\coloneqq\varphi_{\big(\varphi_{A_{1} \rightarrow A_{2}}(t_1)\big)\rightarrow \big(\varphi_{A_{3} \rightarrow A_{4}}(t_2)\big)}(t_3).
\end{displaymath}
Recursively, we can construct the balanced $p$-parameter covariance function. The image of the resulting function is the \emph{balanced} $p$-parameter covariance family.
\end{definition}

\Cref{fig:trees} illustrates the construction of the two pure $p$-parameter covariance functions. The structure can be understood as a tree where the anchor matrices are represented by leaf nodes and pairs of nodes each have one child. {Each child node is associated with a parameter $t_i$.} A \emph{mixed} $p$-parameter covariance function can be obtained by combining balanced and unbalanced covariance functions. In the balanced tree structure, we require the number of anchor matrices to be a power of two.

\begin{figure}[!ht]
  \centering
  \label{fig:trees}\includegraphics[scale=1.5]{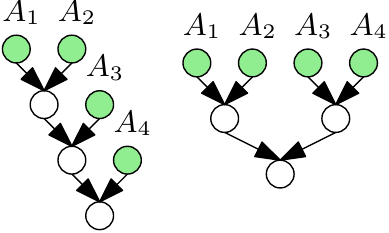}
  \caption{Unbalanced (left) versus balanced (right) trees. Green nodes represents anchor matrices and white circles are combination of one-parameter covariance families.}
\end{figure}

{
Using the tree representation in \Cref{fig:trees}, we see that a covariance family is invariant to swapping the order of the two parents of any node; this follows from the second property listed after~\Cref{def:kernel}. Any other permutation of the anchor matrices will change the image of the multi-parameter covariance function and hence the family. As a specific example, consider the unbalanced two-parameter covariance function $\varphi_{A_{1} \rightarrow A_{2}\rightarrow A_{3}}(t_1,t_2)$ from~\Cref{def:pkernel}. There are six possible orderings of the anchors. Swapping the first two matrices of any ordering does not change the image of the covariance function, but any other permutation does; thus we can describe three different covariance families.
}

As in \Cref{sec:defii}, here we present the natural projection for a generic $p$-parameter covariance function. I-projection and reverse I-projection can be defined analogously. Let $\widehat{C}$ be the sample covariance matrix of the data $\left\{  y_k \right \}_{k=1}^{q}$, and assume that $\widehat{C}$ is full rank.

\begin{problem}[\textbf{Natural projection to a $p$-parameter covariance function}]\label{prob:pmin}
\begin{displaymath}
{\arg\min}_{t_1,\dots,t_p\in (-\infty,\infty)}\; d\big(\varphi_{A_{1} \rightarrow \dots\rightarrow A_{p+1}}(t_1,\dots,t_p),\widehat{C}\big).
\end{displaymath}
\end{problem}


\begin{algorithm}[!h]
\caption{Coordinate descent for an unbalanced $p$-variate covariance function}
\label{algo:coordinate}
\textbf{Input}: $A_1,\dots,A_j,\dots,A_{p+1}\in \Sp$.
\begin{enumerate}
\item Use $\mathbf{t}^{(0)}=[t_1^{(0)},\dots,t_p^{(0)}]=\mathbf{0}$ as initial guess.
\item For $j=1:p$, find \\
$t_j^{(1)}={\arg\min}_{t_j\in (-\infty,\infty)}\; d(\varphi_{A_{1} \rightarrow \dots\rightarrow A_{p+1}}(t_1^{(1)},\dots,t_{j-1}^{(1)},t_j,t_{j+1}^{(0)},\dots,t_p^{(0)}),\widehat{C}).$
\item Repeat step 2 for $N$ iterations until convergence.
\end{enumerate}
\textbf{Return}: The approximate minimizer is then $\varphi_{A_{1} \rightarrow \dots\rightarrow A_{p+1}}(t_1^{(N)},\dots,t_p^{(N)})$.
\end{algorithm}

To solve \Cref{prob:pmin}, we propose \Cref{algo:coordinate} based on coordinate descent. The objective function of \Cref{prob:pmin} is convex with respect to the first variable (\Cref{lem:convexity}). However, it is not necessarily convex in other directions. Therefore, \Cref{algo:coordinate} is not guaranteed to converge to a global minimum. By construction, however, the distance obtained via the algorithm is non-increasing as we increase the size of the family $p$ or the number of iterations $N$. In addition to providing a matrix within the family, the algorithm outputs the corresponding parameter values $t_1^{*},\dots,t_p^{*}$. In practice, as with many coordinate descent algorithms, we find that this simple approach performs well.

\Cref{algo:coordinate} can also be extended to the balanced $p$-variate covariance function. To do so, it suffices to define an order for Step $2$. The simplest strategy is  first to minimize with respect to each parameter connecting each pair of parents (cf.~\Cref{fig:trees}) and subsequently each pair of children, descending through the hierarchy. The same process can be applied for the mixed covariance function.

\section{Case study: hydraulic head in an aquifer}\label{sec:aquifer}

In this section, we use an example from groundwater hydrology to understand the capabilities of the covariance families and estimation methods developed above. We consider a simple model of the hydraulic head in an aquifer, illustrated in \Cref{fig:aquifer}, where the spatially heterogeneous permeability is modeled as a random field. We are interested in estimating the covariance of the resulting hydraulic head, across multiple points in the spatial domain. The stochastic model for the permeability field, which reflects various scenarios of geostatistical knowledge, directly impacts the covariance of the hydraulic head.





\begin{figure}[!ht]
  \centering
  \includegraphics[scale=1]{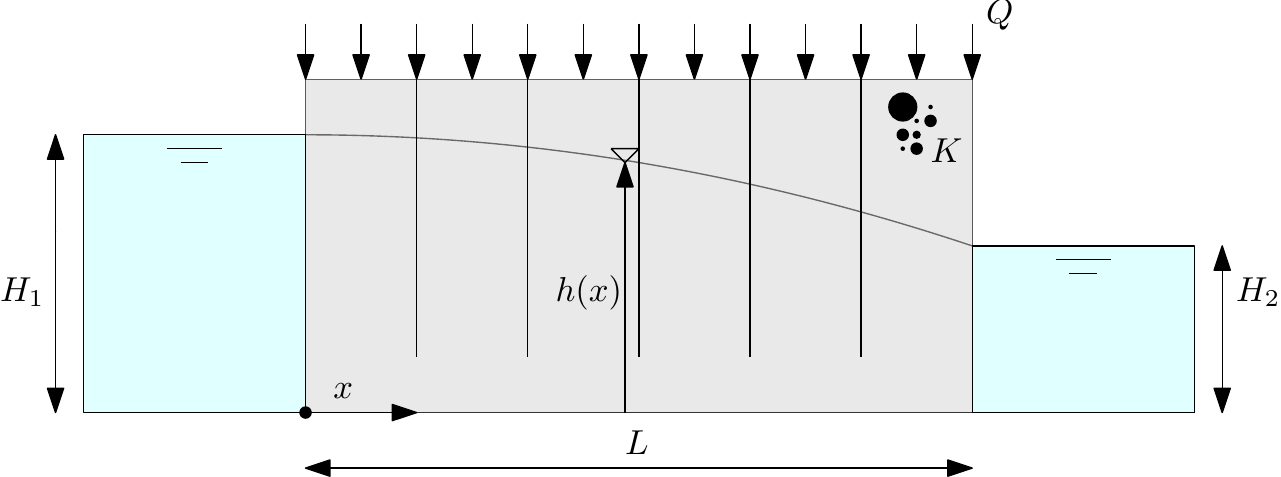}
  \caption{Illustration of the considered aquifer.}
  \label{fig:aquifer}
\end{figure}

\subsection{Analytical model}

The hydraulic head $h$ in the aquifer can be modeled by a one-dimensional Poisson equation with a stochastic permeability coefficient $\kappa$,
\begin{equation}\label{pde}
\dfrac{\partial}{\partial x}\left(\kappa(x,\omega)\dfrac{\partial h(x,\omega)}{\partial x}\right) + Q(x) = 0,\quad x\in \Omega=[0,L],
\end{equation}
where the source term is uniform $Q(x) = Q = 0.02$ and the boundary conditions are Dirichlet:
\begin{equation*}
h(0) = H_1=50,\;h(L) = H_2=20.
\end{equation*}
The permeability field $\kappa(x, \omega)$ is here defined as the exponential of a Gaussian process on $[0,L]$ with constant mean $\mu(x) = 1$ and covariance kernel,
\begin{equation*}
C(x,x') = \sigma^2\exp\left(-\dfrac{1}{p}\left(\dfrac{\vert x-x' \vert}{l} \right)^p\right).
\end{equation*}
In the examples below, we will use $p=2$ and $L=100$, with various values of the correlation length $l$ and variance $\sigma^2$ as indicated.

\subsection{Construction of the covariance family}\label{sec:cs_construction}
For any realization of the permeability field $\kappa$, we solve the equation above using a second-order finite difference scheme. By drawing many realizations of the log-permeability from the Gaussian process defined above, we can use Monte Carlo simulation to construct a sample estimate of the covariance of {$\{ h(x_i,\omega)\}_{i=1}^n$}, taken at $n=20$ equally spaced points $\{x_i\}$ on the domain. We do so for two different values of the correlation length $l$, termed $l_1$ and $l_2$, fixing $\sigma^2 = 0.3$, and build a one-parameter covariance family using the corresponding covariance matrices of {$h$} (called $A_1$ and $A_2$) as anchors.

In \Cref{fig:casestudy1}, we show an initial comparison of the one-parameter geodesic covariance family $\varphi_{A_1 \rightarrow A_2}(t)$ with the ``flat'' covariance family given by $t \mapsto (1-t)A_1 + tA_2$. We plot the distance from each point in the family to another given matrix ($A_3$, obtained with a log-permeability correlation length of $l_3$) for two cases: one where $A_1$ is closer to $A_2$ (left) and the other where $A_1$ is farther from $A_2$ (right). We see that if the two anchors      are far apart, the natural distance from $A_3$ to the one-parameter \emph{flat} covariance family loses convexity; moreover, it is not well defined for the entire real line, as the covariance matrices in the family lose rank for certain values of $t$. In contrast, the distance to the geodesic covariance family is convex and well defined for all $t \in \mathbb{R}$. In all of these cases, we use a very large number of Monte Carlo samples $(q=10^6)$ to construct $A_1$, $A_2$, and $A_3$ so that sampling error is negligible.

\begin{figure}[!ht]
  \centering
  \includegraphics[scale=.6]{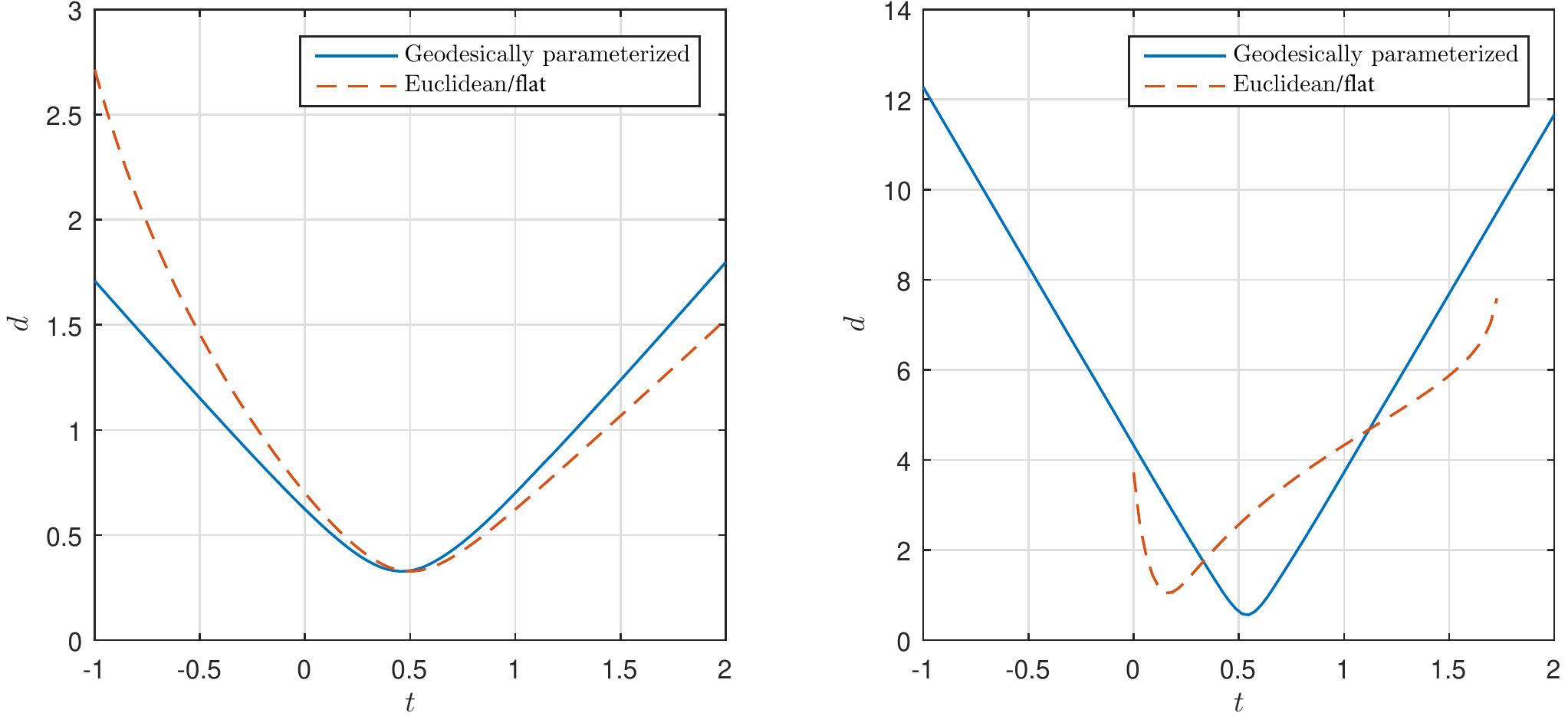}
  \caption{Contrasting the geodesic and ``flat'' one-parameter covariance families, by evaluating the natural distance to a third matrix. Left: when the anchor matrices are close $(l_1=20$, $l_2=30$, and $l_3=25)$ the two families are similar. Right: when the anchors are far apart $(l_1=20$, $l_2=100$, and $l_3=60)$, the natural distance to the flat family is not convex; outside of a certain range, where the red dashed line disappears, it is not even well defined.}

  \label{fig:casestudy1}
\end{figure}

\subsection{Regularization of a solution obtained with a reduced number of Monte Carlo instances}\label{sec:estimation1}

As described above, a standard method for solving \Cref{pde}---e.g., computing the covariance of the solution field $h$---is Monte Carlo simulation. However, this approach converges slowly and can require a significant number of samples to produce an accurate estimate, thus incurring significant computational cost. Indeed, a central concern of forward uncertainty quantification (UQ) is the development of methods to characterize $h(x, \omega)$ with a cost that is much smaller than that of direct Monte Carlo simulation. Yet most UQ approaches focus on solving \Cref{pde} and similar stochastic PDEs only for a \emph{single} specification of the stochastic process $\kappa(x,\omega)$ \cite{lord2014introduction}; if the parameters describing the stochastic inputs change, the problem typically must be re-solved entirely.

Here we explore how to use tailored covariance estimation to solve this ``outer'' problem accurately using a rather small number of Monte Carlo samples. We propose to compute the covariance matrix of the solution (here called $A_3$) for new values of the input correlation length as follows:
\begin{enumerate}
\item Construct a covariance family using related problem instances. In the current example, we build a one-parameter covariance family using the anchor matrices $A_1$ and $A_2$ described in the previous subsection, corresponding to permeability field correlation lengths $l_1=20$ and $l_2=30$. These anchors are obtained with a large number of Monte Carlo samples ($q = 10^6$, {$q/n=5\times10^4$}).
\item Compute the sample covariance matrix $\hat{A}_3$ at the desired new value of the permeability correlation length (here, $l_3 = 25$) using a reduced number of Monte Carlo samples ($q=10^3$, {$q/n=50$}).
\item Project $\hat{A}_3$ to the family, via natural projection, to obtain a covariance matrix estimate $A_3^{*}$ that is ideally closer to the actual solution $A_3$.
\end{enumerate}

\begin{figure}[!ht]
  \centering
  \includegraphics[scale=1]{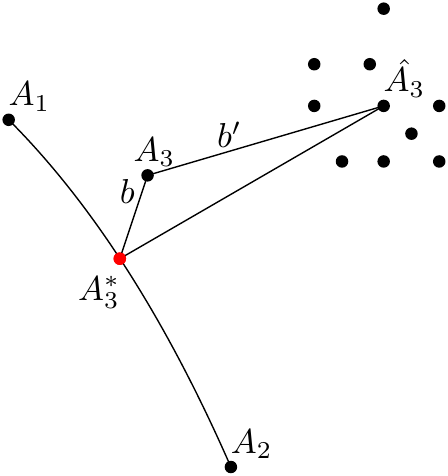}
  \caption{Illustration of regularizing by projection.}
  \label{fig:estimation1}
\end{figure}

\cref{fig:estimation1} illustrates the proposed method. In \Cref{fig:estimation1fig}, we show the impact of this regularization scheme by performing 1,000 instances of the numerical experiment. In each instance, we repeat steps 2 to 3 above, i.e., we keep the anchors $A_1$ and $A_2$ fixed and only recalculate the noisier sample covariance $\hat{A}_3$.
On average, the natural distance $b'$ from the initial sample covariance estimate $\hat{A}_3$ to the true covariance matrix $A_3$ is 0.77 units. The average distance $b$ from the projected matrix $A_3^{*}$ to the true covariance is 0.07. The average reduction in error (i.e., $b'/b$ averaged over problem instances), which can be understood as a regularization ratio, is 11.8.
The blue histogram of distances has a hard minimum at 0.04, which reflects limitations of the covariance family: the true solution $A_3$ is close to the geodesic but does not actually belong to it.

\begin{figure}[!ht]
  \centering
  \includegraphics[scale=.52]{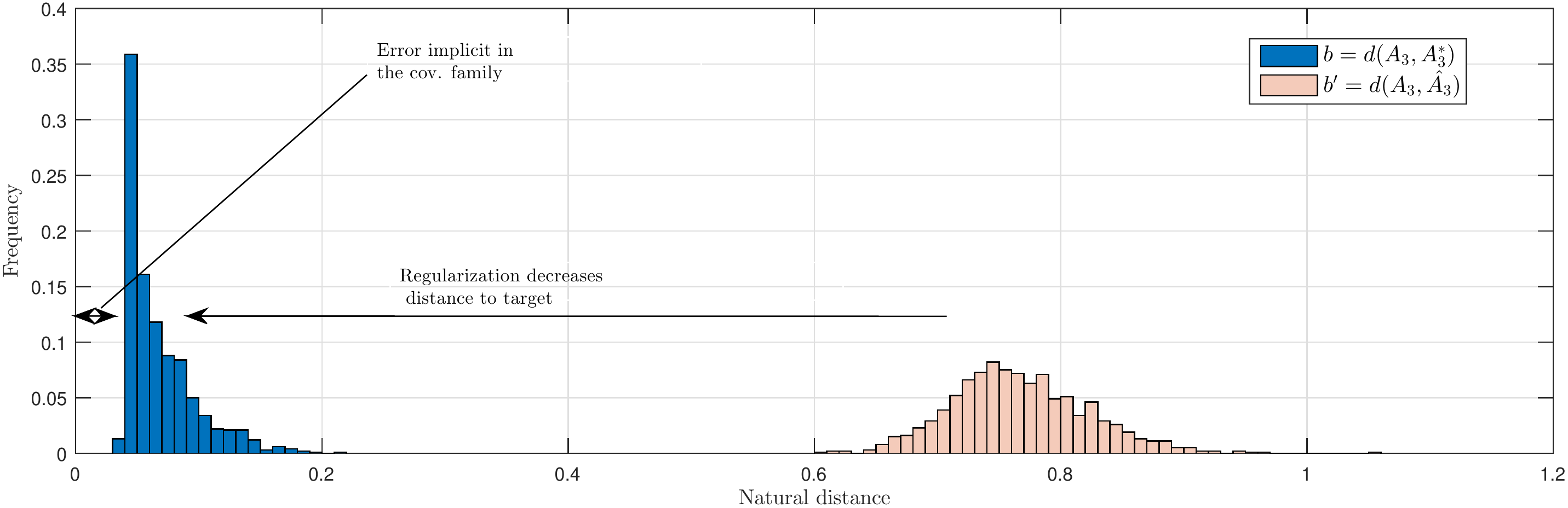}
  \caption{Using 1,000 solutions obtained via the process outlined in \Cref{sec:estimation1}, we plot a histogram of the initial distance $b'$ to the true covariance matrix and a histogram of the distance $b$ once the sample covariance is regularized via natural projection (cf.~\Cref{fig:estimation1}). The averages of $b'$ and $b$ are 0.77 and 0.07, respectively. On average, we reduce the error by a regularization ratio of 11.8.}
  \label{fig:estimation1fig}
\end{figure}

{Setting aside the offline cost of computing the anchors, the computational cost of performing natural projection of $\hat{A}_3$ to the geodesic family is the cost of solving \Cref{prob:min}. This problem is univariate and convex, and thus easily tackled by a variety of methods; here we simply use direct search, which usually requires fewer than 10 iterations to find the optimal $t^{*}$ with an absolute precision of $10^{-4}$. Evaluating the cost function requires solving a symmetric definite generalized eigenvalue problem at each iteration. Generically, these problems have $O(n^3)$ cost for covariance matrices of size $n$, though the complexity may be lower with iterative solvers and problems with particular structure (since the natural distance depends most strongly on the extreme eigenvalues, i.e., those that differ most from one). In this example, the actual cost of solving the eigenproblem is negligible, particularly compared to the cost of drawing Monte Carlo samples to form $\hat{A}_3$. In general, we note that the cost of drawing Monte Carlo samples typically \emph{also} scales with $n$. For instance, each Monte Carlo draw might involve a linear PDE solve (say of $O(n^2)$ complexity) and to maintain a fixed $q/n$, the sample size $q$ must itself increase linearly with $n$.}


\subsection{Regularization of a noisy solution}\label{sec:estimation3}

Projection onto the geodesic family, as illustrated in \Cref{sec:estimation1}, can also be performed using maximum likelihood or I-projection; we find that these approaches yield similar regularization performance for the previous problem.
%
%
Now we consider a more difficult regularization task, where Monte Carlo samples are perturbed with independent and identically distributed realizations of zero-mean Gaussian noise before they are used to construct the sample covariance matrix. This problem mimics a situation where noisy observational data are used to estimate the covariance of the hydraulic head $h$. The anchor matrices $A_1, A_2$ and true covariance matrix $A_3$ are the same as in the previous problem.
\Cref{fig:estimation3fig} shows the regularization ratios $b'/b$ obtained for both natural projection (left) and maximum likelihood estimation (right), at ten different noise magnitudes, each using 500 instances.
On each box, the central mark indicates the median value of $b'/b$, and the bottom and top edges of the box indicate the 25th and 75th percentiles  respectively. The whiskers extend to the most extreme data points not considered outliers, and the outliers are plotted individually using the ``$+$'' symbol. The horizontal axis corresponds to the standard deviation of the Gaussian noise, scaled by $\alpha 0.05 \sqrt{\Tr(A_3)/n}$.

\begin{figure}[!ht]
  \centering
  \includegraphics[scale=0.60]{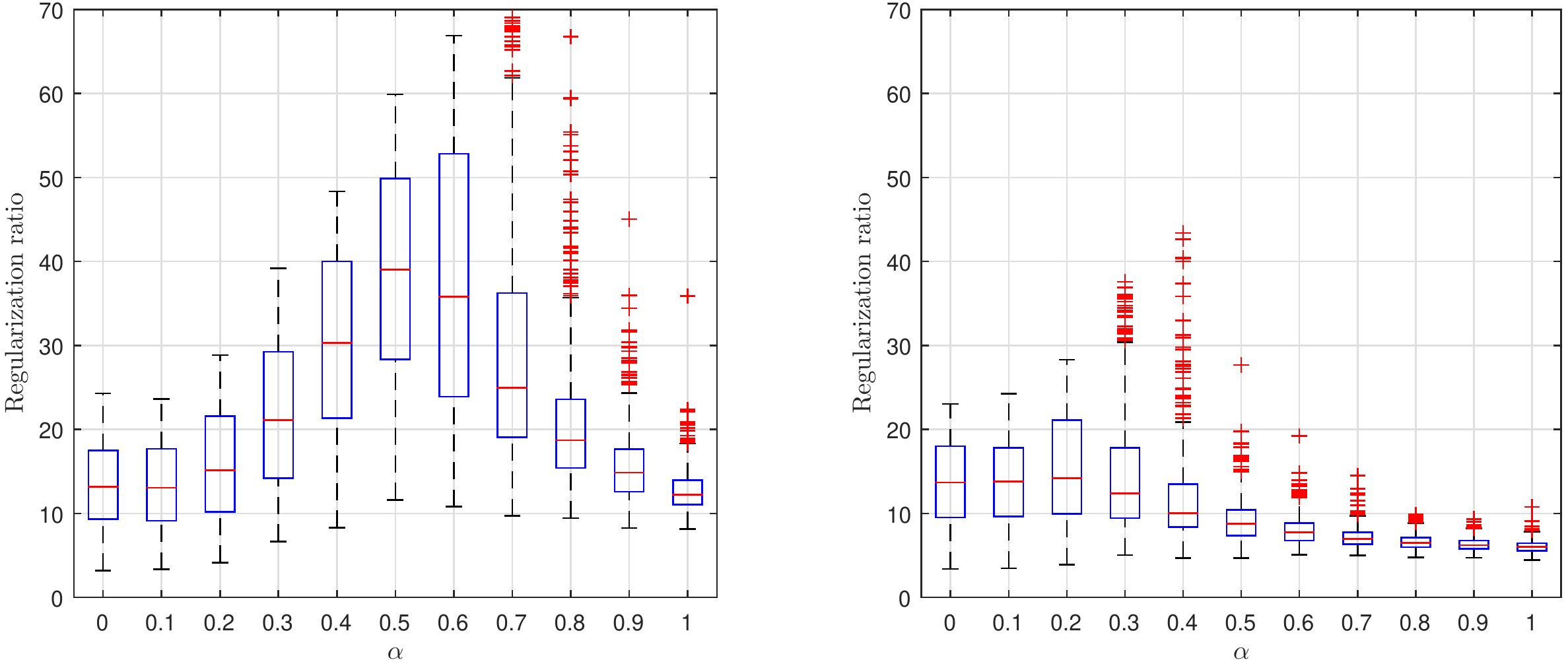}
  \caption{Boxplot of regularization ratios obtained with natural projection (left) and maximum likelihood (right) for increasing magnitudes of noise. Larger values are better. Each box is the result of 500 instances. Standard deviation of the noise perturbations is $\alpha 0.05 \sqrt{\Tr(A_3)/n}$.}
  \label{fig:estimation3fig}
\end{figure}

\Cref{fig:estimation3fig} suggests that maximum likelihood and natural projection perform quite differently in the presence of noise. Both projections yield similar results for $\alpha \leq 0.2$, but natural projection is more robust to noise-corrupted samples for larger noise perturbations. The covariance families for both cases are exactly the same, but maximum likelihood appears less able to identify the closest covariance matrix in the family as the noise magnitude increases. Both regularization methods display an up-down trend with $\alpha$. As $\alpha$ first increases, $b'$ increases while $b$ stays relatively constant; in other words, the sample covariance moves further from the true $A_3$ but the quality of the regularized matrix $A_3^*$ does not deteriorate, and thus we observe larger regularization ratios (much more so for natural projection). As $\alpha$ grows even larger, eventually the distance $b$ starts increasing as well, and thus $b'/b$ begins to fall. The variance of the regularization ratios is somewhat larger for natural projection, but the median ratio obtained with maximum likelihood is generally even smaller than the minimum ratio achieved with natural projection.

As discussed in \Cref{sec:mainres}, the optimality equation for maximum likelihood estimation in the covariance family is a linearized version of  natural projection. It may be that maximum likelihood is more sensitive to noise-corrupted data because it is missing certain higher-order terms. Further exploration of natural projection's ability to overcome noise might employ a perturbation analysis of the generalized eigenvalues of the spectral functions in \Cref{rem:examples1}; we leave this to future work.

\subsection{Performance of proposed algorithm for multi-parametric families}\label{sec:estimation4}
With more than two anchor matrices, one can build a multi-parametric covariance function as described in \Cref{sec:pvariate}.
To illustrate, suppose that we have three anchor covariance matrices $A_1$, $A_2$, and $A_3$, corresponding to solutions of \Cref{pde} for $(l_1=20,\,  \sigma^2_1=0.3)$, $(l_2=30, \, \sigma^2_2=0.3)$, and $(l_3=25, \, \sigma^2_3=0.4)$, respectively. We now solve a problem similar to that in \Cref{sec:estimation1}, but with two parameters. The objective is to approximate the covariance matrix $A_4$, obtained with $(l_4=25, \,  \sigma^2_4=0.35)$. First, we build an unbalanced two-parameter covariance family $\varphi_{A_{1} \rightarrow A_{2}\rightarrow A_{3}}(t_1,t_2)$ (see \Cref{def:pkernel}).
Then, we project an approximation $\hat{A}_4$ of $A_4$, obtained with 1,000 Monte Carlo simulations of \Cref{pde}, onto the family. Repeating this experiment with 1,000 independent realizations of $\hat{A}_4$, we obtain an average regularization ratio of 6.5. These results are illustrated in \Cref{fig:estimation1fig4}.

\begin{figure}[!ht]
  \centering
  \includegraphics[scale=.64]{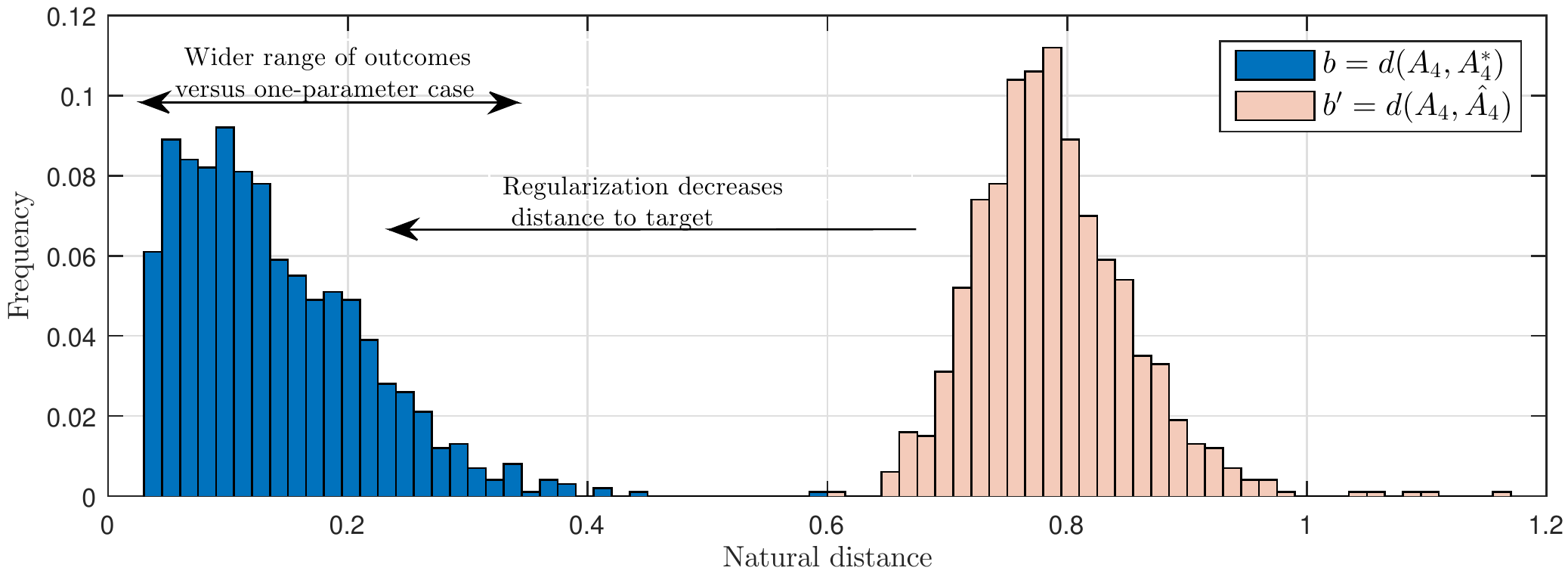}
  \caption{Multi-parametric case: using 1,000 solutions obtained via the process outlined in \Cref{sec:estimation4}, we plot a histogram of the initial distance $b'$ to the true covariance matrix and a histogram of the distance $b$ obtained after natural projection. The average distances are 0.79 and 0.14, respectively. On average, we reduce the error by a regularization ratio of 6.5.}
  \label{fig:estimation1fig4}
\end{figure}

To perform natural projection onto this multi-parametric family, we used \Cref{algo:coordinate}. To illustrate the performance of this algorithm, the colored contours in \Cref{fig:algo} (left) show the distance to $\hat{A}_4$ as a function of the covariance family's parameters $t_1, t_2$. Iterations of the algorithm $1,\dots,N$ are marked with green triangles: intermediate iterations are unfilled triangles, and the filled triangle denotes the final point. We say that convergence is achieved when consecutive values of $t_1$ and $t_2$ each do not differ by more than $10^{-4}$. In this particular example, convergence requires four iterations, three of which fall inside the perimeter of \Cref{fig:algo} (left).

In \Cref{fig:algo} (right), color contours illustrate the distance from the family to the true covariance matrix $A_4$, with iterations of the coordinate descent algorithm  again overlaid. As the algorithm progresses, the distance to $A_4$ does not necessarily decrease; as expected, the minimization is blind to $A_4$.
Indeed, the goal of \Cref{algo:coordinate} is to obtain the minimizer of $d\bigl(\hat{A}_4,\varphi_{A_{1} \rightarrow A_{2}\rightarrow A_{3}}(t_1,t_2)\bigr)$ (green triangle), which is generally not the same as the minimizer of $d\bigl(A_4,\varphi_{A_{1} \rightarrow A_{2}\rightarrow A_{3}}(t_1,t_2)\bigr)$ (red square). The natural distance between the covariance matrices corresponding to these two minimizers is 0.08 units. This distance reflects the fact that $\hat{A}_4$ is noisy, and thus its best approximation in the family is not the best approximation of $A_4$. Separately, the limit of the two-parameter covariance family's ability to represent $A_4$ is captured by the value of the contour in \Cref{fig:algo} (right) at the red square, which is 0.03 units.

\begin{figure}[!ht]
  \centering
  \includegraphics[scale=.6]{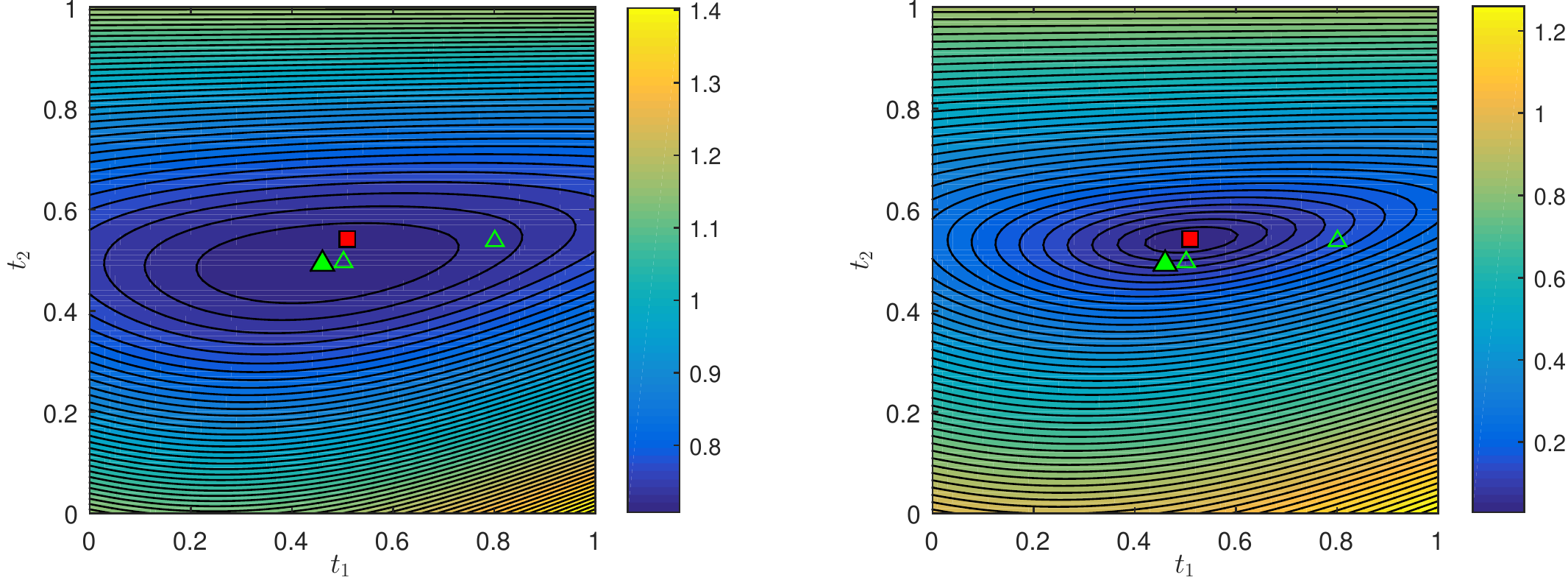}
  \caption{Multi-parametric regularization for one instance of $\hat{A}_4$, described in \Cref{sec:estimation4}. Left panel: contours/colors represent the value of $d\bigl(\hat{A}_4,\varphi_{A_{1} \rightarrow A_{2}\rightarrow A_{3}}(t_1,t_2)\bigr)$ and triangles show iterations of \Cref{algo:coordinate} until convergence. Red square is the minimizer of $d\bigl(A_4,\varphi_{A_{1} \rightarrow A_{2}\rightarrow A_{3}}(t_1,t_2)\bigr)$. Right panel: contours/colors represent $d\bigl({A_4},\varphi_{A_{1} \rightarrow A_{2}\rightarrow A_{3}}(t_1,t_2)\bigr)$.}
  \label{fig:algo}
\end{figure}

\section{Conclusions}\label{sec:conclu}
We have proposed a framework for building expressive and problem-tailored parametric covariance families by connecting representative ``anchor'' covariance matrices through geodesics. The building block of the framework is the one-parameter covariance family, corresponding to a single geodesic. These geodesics may be combined to yield multi-parameter covariance families.
Given some new data, one can then choose the most appropriate member of such a family by minimizing the natural distance (on the manifold of symmetric positive-definite matrices) to the sample covariance matrix of the data. We call this notion natural projection. Unlike maximum likelihood estimation (reverse I-projection) or I-projection {within the family}, natural projection is consistent with the notion of distance employed to build the covariance family. We elucidate the differences among these estimation techniques and show that I-projection and reverse I-projection can be seen as linearizations of the natural projection.

We also illustrate the advantages of geodesic covariance families and natural projection in several numerical experiments. Analogous covariance families that do not employ the geodesic structure may lose rank, and the distance from such a ``flat'' family to another matrix is in general not convex---especially if the anchors are far apart. The geodesic families avoid these difficulties. When performing parameter estimation within the geodesic family, maximum likelihood and natural projection provide similar results in the absence of noise. If the data are corrupted by noise, however, then natural projection is better able to regularize the solution.


{Given a covariance family,} the choice between natural projection and maximum likelihood estimation {within the family} also depends on the number of data points $q$ and the size of the matrices $n$. When $q<n$, the natural distance cannot be used because the sample covariance matrix will be rank deficient, and thus will not belong to the manifold of symmetric positive-definite matrices. On the other hand, when $q>n$, we suggest that it is preferable to use natural projection because it is consistent with the geodesic construction of the family (with a cost function that is a proper notion of distance) and because it does not require assigning a distribution to the data. Moreover, it has superior noise rejection properties. {If the sample covariance matrix is close to the family, however,} we have also shown that minimizing the natural distance, maximizing the likelihood, and performing I-projection within the family coincide up to second order. {Moreover, as $q/n \to \infty$, the sample covariance matrix itself becomes a good representation of the true (population) covariance, and in this limit, the practical need to project to any parametric family diminishes. Nonetheless, we have demonstrated numerically that even for $q/n \approx 50$, regularization of the sample covariance matrix via projection can yield significant reductions in error. And if the population covariance matrix is contained within the geodesic family, then the natural projection of the sample covariance yields a consistent estimator of the true covariance, as one would certainly desire.}

A natural extension of this work is to weaken the role of the anchor matrices: not to project directly to a parametric family defined by the anchors, {but rather to seek only ``closeness'' to one or more anchors, where the degree of closeness might depend on the quality of the sample covariance matrix.} A popular strategy along these lines is linear shrinkage, which consists in selecting a linear combination of the sample covariance matrix and a single reference covariance matrix (often chosen to be the identity). In particular, linear shrinkage seeks the combination that is closest (in expected Frobenius distance) to the true covariance; thus both the loss and the effective covariance family are flat. The obvious challenge is that the true covariance matrix is not known. {In the spirit of the present paper, future work could develop a geodesic version of shrinkage.} We expect that such a nonlinear shrinkage (cf.\ \cite{ledoit2018optimal}) could extend some of the favorable properties of the geodesic framework to the non-parametric case.

{Another promising application of the covariance families developed here may lie in hierarchical Bayesian modeling, e.g., for inverse problems. Specifically, we suggest that geodesically parameterized covariance matrices could be used to describe particularly flexible classes of prior models, where the covariance family's parameters $(t_1, \ldots, t_p)$ may serve as hyperparameters---rather than the correlation/scale/smoothness parameters of standard covariance kernels. The parameters of the covariance family could then be inferred either in a fully Bayesian formulation or via an empirical Bayesian approach. A different possibility is to use geodesic families to continuously interpolate among the posterior covariance matrices that follow from particular prior choices. }

\appendix
\section{Technical results}\label{sec:technical}

\begin{lemma}[\textbf{Spectral function minimization}]\label{lem:spectral}
%
%
    Let $F=f\circ\lambda$ be a spectral function and let $X(t) \coloneqq \widehat{C}^{-\frac{1}{2}}\varphi_{A_{1} \rightarrow A_{2}}(t)\widehat{C}^{-\frac{1}{2}}=M\Lambda^tM^{\top}$, where $\varphi_{A_{1} \rightarrow A_{2}}(t)$ is the geodesic defined in \Cref{eq:geodesic}, $M = \widehat{C}^{-\frac{1}{2}} A_{1}^{\frac{1}{2}}U$, and $\widehat{C}$ is full rank.
    Minimizing $F\bigl(X(t)\bigr)$ over $t$ is equivalent to finding $t^{+}$ such that:\footnote{$\frac{d f}{d \lambda}$ is our shorthand notation for $ \diag\bigg(\frac{d f}{d \lambda_1}, \dots, \frac {d f}{d \lambda_n}\bigg).$}
\begin{displaymath}
\Tr\left(V(t^{+})\Bigg( \frac{d f}{d \lambda}\Bigg|_{\lambda \big{(}X(t^{+})\big{)}} \Bigg)V(t^{+})^{\top}M\Lambda^{t^{+}} \log_m(\Lambda) M^{\top}\right)=0,
\end{displaymath}
where $V(t)\Sigma(t)V(t)^{\top}$ is an orthonormal eigendecomposition of $X(t)$.

\end{lemma}

\begin{proof}
  Notice that the generalized eigenvalues of the pencil $\big(\varphi_{A_{1} \rightarrow A_{2}}(t),\widehat{C}\big)$ are the eigenvalues of $X(t)$. Since this is an unconstrained minimization problem, the idea is to impose the following condition:
  \begin{displaymath}
    \frac{d F\big(X(t)\big)}{d t}\Bigg|_{X(t^{+})}=0.
  \end{displaymath}
The difficulty here is that $F\big(X(t)\big)=f\circ\lambda\circ X(t)$, where $X(t)$ is defined as in the present Lemma, $\lambda$ is the function that extracts the eigenvalues of a given matrix, and $f$ is a mapping from these eigenvalues to $\mathbb{R}$. Applying the chain rule and \cite[Theorem 1.1]{lewis1996derivatives}, we obtain:
\begin{displaymath}
 \frac{d F\big(X(t)\big)}{d t}=\Tr\left(V(t) \left . \frac{d f}{d \lambda} \right |_{\lambda \big{(}X(t)\big{)}} V^{\top}(t)\frac{d X(t)}{dt}\right),
\end{displaymath}
where $ \frac{d X(t)}{dt}=M\Lambda^t \log_m(\Lambda) M^{\top}$.
\end{proof}

\medskip

The next three proofs of results from \Cref{sec:mainres} follow similar strategies. We first use \Cref{lem:spectral} for the corresponding spectral function in \Cref{rem:examples1}, and then we derive the orthogonality condition.

\begin{proof}[Proof of \Cref{thm:solgeom}]
We start with (cf.~\Cref{eq:natdist}):
\begin{displaymath}
f(\lambda_1, \dots, \lambda_n)=\sqrt{\sum_{k=1}^{n}\log^2\lambda_{k}}.
\end{displaymath}

After omitting the square root, the derivative is:
\begin{displaymath}
\frac{d f(\lambda)}{d \lambda_k}=\frac{2 \log_m\Big(\lambda_k \big(X(t)\big)\Big)}{\lambda_k\big(X(t)\big)}=\Big( 2 \Sigma^{-1}(t)\log_m \big(\Sigma(t)\big) \Big)_{(k,k)},
\end{displaymath}
where $\Sigma(t)$ is defined in the preceding \Cref{lem:spectral} .We have to find a $t^*$ such that:
\begin{equation}\label{eq:firstder}
\frac{d F\big(X(t)\big)}{d t}\Bigg|_{X(t^*)}=2\Tr\Big(V(t^*) \Sigma^{-1}(t^*)\log_m \big(\Sigma(t^*)\big)V^{\top}(t^*)M\Lambda^{t^*} \log_m(\Lambda) M^{\top}\Big)=0.
\end{equation}
Now, notice that by construction $M^{\top} V(t)\Sigma^{-1}(t)=\Lambda^{-t}M^{-1}V(t)$, and applying the cyclical property of the trace:
\begin{displaymath}
\Tr\Big(M^{-1}V(t^*)\log_m \big(\Sigma(t^*)\big)V^{\top}(t^*)M\Lambda^{t^*} \log_m(\Lambda)\Lambda^{-t^*}\Big)=0.
\end{displaymath}
Since diagonal matrices commute:
\begin{displaymath}
\Tr\Big(\log_m \big(M^{-1}V(t^*)\Sigma(t^*)V^{\top}(t^*)M\big) \log_m(\Lambda)\Big)=0.
\end{displaymath}
After that, we obtain:
\begin{displaymath}
\Tr\big(\log_m (\Lambda^{t^*}M^{\top} M) \log_m(\Lambda)\big)=0.
\end{displaymath}
The first part of the proof (cf.\;\Cref{eq:soldistance}) is concluded after realizing that by construction $M^{\top} M =U^{\top}A_1^{\frac{1}{2}}\widehat{C}^{-1}A_1^{\frac{1}{2}}U$ and:
\begin{equation}
    \Tr\big(\log_m(Z\Lambda^{-t^{*}})\log_m(\Lambda)\big)=0. \label{eq:soldistance}
\end{equation}

Then note that \Cref{eq:minopt} can be rewritten as:
\begin{displaymath}
\Tr\left (\log_m \left (R \left (-\frac{t^{*}}{2} \right ) A_1^{-\frac{1}{2}}\widehat{C}A_1^{-\frac{1}{2}}R \left (-\frac{t^{*}}{2} \right ) \right ) \log_m\left (R \left (-\frac{t^{*}}{2} \right )R(1+t^{*})R \left (-\frac{t^{*}}{2} \right )\right ) \right )=0,
\end{displaymath}
which is equivalent to \Cref{eq:ortho}.
\end{proof}

\begin{proof}[Proof of \Cref{thm:sollike}]
For reverse I-projection, we start with (cf.\;\Cref{eq:likelihood}):
\begin{displaymath}
f(\lambda_1, \dots, \lambda_n)=\sum_{k=1}^{n}\frac{\lambda_{k}^{-1}+\log\lambda_{k}-1}{2}.
\end{displaymath}
Then:
\begin{displaymath}
\frac{d f(\lambda)}{d \lambda_k}=-\frac{1}{2\lambda_k\big(X(t)\big)^{2}}+\frac{1}{2\lambda_k\big(X(t)\big)}=\frac{1}{2}\big(- \Sigma^{-2}(t)+\Sigma^{-1}(t) \big)_{(k,k)}.
\end{displaymath}
Similarly to the previous case, now we have to find a $\hat{t}$ such that:
\begin{displaymath}
\Tr\Big(V(\hat{t})\big(-\Sigma^{-2}(\hat{t})+\Sigma^{-1}(\hat{t})\big)V^{\top}(\hat{t})M\Lambda^{\hat{t}} \log_m(\Lambda) M^{\top}\Big)=0.
\end{displaymath}
Applying $M\Lambda^tM^{\top}=V(t)\Sigma(t)V(t)^{\top}$ and the cyclical property of the trace, we obtain:
\begin{displaymath}
\Tr\big(-M^{-T}\log_m(\Lambda)\Lambda^{-\hat{t}}M^{-1}+\log_m(\Lambda)\big)=0.
\end{displaymath}
The result follows after applying $M=\widehat{C}^{-\frac{1}{2}}A_{1}^{\frac{1}{2}}U$, that is:

\begin{equation}\label{eq:solreviproj}
\Tr\big((Z\Lambda^{-\hat{t}}-\Id)\log_m(\Lambda)\big)=0.
\end{equation}
The orthogonality condition is obtained as in \Cref{thm:solgeom}.
\end{proof}

\begin{proof}[Proof of \Cref{thm:soliproj}]
For I-projection, we start with (cf.\;\Cref{eq:kl_div}):
\begin{displaymath}
f(\lambda_1, \dots, \lambda_n)=\sum_{k=1}^{n}\frac{\lambda_{k}-\log\lambda_{k}-1}{2}.
\end{displaymath}
Then:
\begin{displaymath}
\frac{d f(\lambda)}{d \lambda_k}=\frac{1}{2}-\frac{1}{2\lambda_k\big(X(t)\big)}=\frac{1}{2}\big(\Id-\Sigma^{-1}(t) \big)_{(k,k)}.
\end{displaymath}
Similarly to the previous case, now we have to find a $\check{t}$ such that:
\begin{displaymath}
\Tr\Big(V(\check{t})\big(\Id-\Sigma^{-1}(\check{t})\big)V^{\top}(\check{t})M\Lambda^{\check{t}} \log_m(\Lambda) M^{\top}\Big)=0.
\end{displaymath}
Applying $M\Lambda^tM^{\top}=V(t)\Sigma(t)V(t)^{\top}$ and the cyclical property of the trace, we obtain:
\begin{displaymath}
\Tr\big(M^{\top}M\Lambda^{\check{t}}\log_m(\Lambda)-\log_m(\Lambda)\big)=0.
\end{displaymath}
The result follows after applying $M=\widehat{C}^{-\frac{1}{2}}A_{1}^{\frac{1}{2}}U$, that is:
\begin{equation}\label{eq:soliproj}
\Tr\big((\Lambda^{\check{t}}Z^{-1}-\Id)\log_m(\Lambda)\big)=0.
\end{equation}
The orthogonality condition is obtained as in \Cref{thm:solgeom}.
\end{proof}

\begin{lemma}[\textbf{Equivalence of likelihood maximization and reverse I-projection}]\label{lem:convergence}
Let $p_Y(\cdot \, ; t)$ be a Gaussian density on $\mathbb{R}^n$ centered at zero, with covariance $\varphi_{A_{1} \rightarrow A_{2}}(t)$. Let $y_1,\ldots,y_q \stackrel{iid}{\sim} p_Y(\cdot \, ; \bar{t})$ for some $\bar{t} \in \mathbb{R}$, such that the sample covariance matrix $\widehat{C}=\frac{1}{q}\sum_{i=1}^q y_iy_i^{\top}$ is full rank. Maximizing the log-likelihood $\sum_{i=1}^{q} \log p_{Y}(y_i;t) $ with respect to $t$ is equivalent to minimizing the KL divergence $\infdiv[\Big]{N(0,\widehat{C})}{N\big(0,\varphi_{A_{1} \rightarrow A_{2}}(t)\big)}$.

\end{lemma}
\begin{proof}

Each observation $y_i$ has the following density:
\begin{displaymath}
p_Y(y_i;t)=\frac{1}{(2\pi)^{n/2}\sqrt{|\varphi_{A_{1} \rightarrow A_{2}}(t)|}}\exp\left(-\frac{1}{2}y_i^{\top}\varphi_{A_{1} \rightarrow A_{2}}^{-1}(t)y_i\right).
\end{displaymath}
The joint log-likelihood can be expressed as:
\begin{displaymath}
\log \prod_{i=1}^q p_Y(y_i;t)=\sum_{i=1}^q \left(-\log^{n/2}2\pi-\frac{1}{2}\log|\varphi_{A_{1} \rightarrow A_{2}}(t)| - \frac{1}{2}y_i^{\top}\varphi_{A_{1} \rightarrow A_{2}}^{-1}(t)y_i \right).
\end{displaymath}
Now notice that $\varphi_{A_{1} \rightarrow A_{2}}(t)=A_{1}^{\frac{1}{2}}U\Lambda^{t}U^{\top}A_{1}^{\frac{1}{2}}$, $|\varphi_{A_{1} \rightarrow A_{2}}(t)|=|A_1||\Lambda^t|$, and ignore the constant terms.
Since this is an unconstrained and concave problem, the extremum $t^+$ can be found by setting the derivative of the function to zero, that is:
\begin{equation}\label{eq:derlikelihood}
 q \Tr\big(\log_m(\Lambda)\big)-\sum_{i=1}^qy_i^{\top}A_1^{-1/2}U\Lambda^{-t^+}\log_m(\Lambda) U^{\top}A_1^{-1/2}y_i=0.\end{equation}
In matrix form, the above expression reads:
\begin{equation}\label{eq:maxlikelihood}
\Tr\big(\widehat{C}A_1^{-\frac{1}{2}}U\Lambda^{-t^{+}}\log_m(\Lambda)U^{\top}A_{1}^{-\frac{1}{2}}-\log_m(\Lambda)\big)=0.
\end{equation}
The proof is concluded after realizing that the last expression is precisely \Cref{eq:solreviproj}.
\end{proof}

\begin{lemma}[\textbf{Convexity of the distance function}]\label{lem:convexity}
The distance function $d(\varphi_{A_{1} \rightarrow A_{2}}(t),\widehat{C})$ is convex in $t$. Therefore, \cref{prob:min} is an unconstrained convex minimization problem.
\end{lemma}
\begin{proof}
The strategy is to take the derivative of \Cref{eq:firstder} and realize that it is non-negative:
\begin{displaymath}
\frac{d F\big(X(t)\big)}{d t}=2\Tr\big(\log_m(M\Lambda^{t}M^{\top})\log_m(M\Lambda M^{-1})\big).
\end{displaymath}
Recall that:
\begin{displaymath}
\frac{d\log_{m} X(t)}{d t}=\int_0^1\big((X(t)-\Id)s+\Id\big)^{-1}\left(\frac{d X(t)}{d t} \right)\big((X(t)-\Id)s+\Id\big)^{-1}ds.
\end{displaymath}
Notice that in our case, $X(t)=M\Lambda^t M^{\top}$ and:
\begin{displaymath}
\frac{d X(t)}{d t} =M\Lambda^t\log_m(\Lambda)M^{\top}.
\end{displaymath}
Performing an orthonormal eigendecomposition of the form $X(t)=V(t)\Sigma V(t)^{\top}$, the other main part of the integrand reads:

\begin{displaymath}
\big((X(t)-\Id)s+\Id\big)^{-1}=\Big(\big(V\Sigma V^{\top}-\Id\big)s+\Id\Big)^{-1}=V\big(s\Sigma +\Id(1-s)\big)^{-1}V^{\top} \coloneqq VJV^{\top},
\end{displaymath}
where for easy presentation, we omit the explicit dependence of $V$, $\Sigma$, and $J$ on $t$. $J$ is clearly positive since $\Sigma$ is and $s\in[0,1]$.

The second derivative can be expressed as:
\begin{displaymath}
\frac{d^2 F\big(X(t)\big)}{dt^2}=
2\int_0^1 \Tr\big(VJV^{\top} M\Lambda^t\log_m(\Lambda)M^{\top} VJV^{\top} M\log_m(\Lambda)M^{-1}\big)ds.
\end{displaymath}
It suffices to show that the trace is positive for $s\in[0,1]$.
Using the equality $V^{\top}M\Lambda^t=\Sigma V^{\top} M^{-T}$ from the eigendecomposition and the cyclical property of the trace:
\begin{eqnarray*}
\frac{d^2 F\big(X(t)\big)}{2dt^2} & = &
\int_0^1 \Tr\big(VJ\Sigma V^{\top} M^{-T}\log_m(\Lambda)M^{\top} VJV^{\top} M\log_m(\Lambda)M^{-1}\big)ds \\
& = & 2\int_0^1 \Tr\big((J\Sigma)^{\frac{1}{2}} V^{\top} M^{-T}\log_m(\Lambda)M^{\top} VJ^{\frac{1}{2}}J^{\frac{1}{2}}V^{\top} M\log_m(\Lambda)M^{-1}V(J\Sigma)^{\frac{1}{2}}\big)ds \\
& = & 2\int_0^1 \Tr(KK^{\top})ds>0,
\end{eqnarray*}
where $K=(J\Sigma)^{\frac{1}{2}} V^{\top} M^{-T}\log_m(\Lambda)M^{\top} VJ^{\frac{1}{2}}$.
\end{proof}


\begin{lemma}[\textbf{Convexity of the KL divergence function}]\label{lem:convexconcave}
The following KL divergences $\infdiv[\Big]{N(0,\widehat{C})}{N\big(0,\varphi_{A_{1} \rightarrow A_{2}}(t)\big)}$ and $\infdiv[\Big]{N\big(0,\varphi_{A_{1} \rightarrow A_{2}}(t)\big)}{N(0,\widehat{C})}$ are convex functions of $t$. Therefore, \cref{prob:reviproj,prob:iproj} are unconstrained convex minimization problems.
\end{lemma}
\begin{proof}
It suffices to evaluate the second derivative and conclude that it is always positive.
Taking derivatives of \Cref{eq:solreviproj}, we obtain:
\begin{displaymath}
\Tr\big(Z\Lambda^{-t}\log_m^2(\Lambda)\big)>0, \ \forall t.
\end{displaymath}
Similarly, taking derivatives of \Cref{eq:soliproj}, we obtain:
\begin{displaymath}
\Tr\big(\Lambda^{t}Z^{-1}\log_m^2(\Lambda)\big)>0, \ \forall t.
\end{displaymath}
\end{proof}

\begin{proposition}[{\textbf{Consistency with a perfect family}}]\label{lem:consistency}
{If the true covariance matrix $A$ is a member of the geodesic covariance family $\varphi_{A_{1} \rightarrow A_{2}}(t)$, natural projection of the sample covariance matrix yields a consistent estimator of $A$.}
\end{proposition}
\begin{proof}

{Without loss of generality, let the observations $(y_i)_{i=1}^q$ be zero-mean random vectors drawn i.i.d.\ from a distribution with covariance matrix $A$, define the sample covariance matrix as $\widehat{A}_{q} \coloneqq \frac{1}{q}\sum_{i=1}^q y_i y_i^{\top}$, and let $A_1 = A$. The natural projection of $\widehat{A}_{q}$ into $\varphi_{A_{1} \rightarrow A_{2}}$ is denoted as $A_{q}^{*}=\varphi_{A_{1} \rightarrow A_{2}}(t_{q}^{*})$, where $t_q^{*} = \argmin_t d(\varphi_{A_{1} \rightarrow A_{2}}(t), \widehat{A}_q)$.
Define $\widehat{Q}_q(t) \coloneqq d\bigl(\varphi_{A_{1} \rightarrow A_{2}}(t),\widehat{A}_{q}\bigr)$ and $Q(t) \coloneqq d\bigl(\varphi_{A_{1} \rightarrow A_{2}}(t),{A}\bigr)$. By \Cref{lem:convexity}
and \Cref{lem:idem}, we have that $Q(t)$ is uniquely minimized at $t=0$, which is clearly in the interior of $(-\infty, \infty)$, and that $\widehat{Q}_q(t)$ is convex for any $q$. Now we show that
$\widehat{Q}_q(t)$ converges in probability to $Q(t)$ for any $t\in \mathbb{R}$:
\begin{displaymath}
\lim_{q \to \infty} \mathbb{P}\bigl(|Q_q(t)-Q(t)|>\epsilon\bigr) \leq \lim_{q \to \infty} \mathbb{P}\bigl(d(A,\widehat{A}_{q})>\epsilon\bigr)=\lim_{q \to \infty} \mathbb{P}\left ({\sum_{k=1}^{n}\log^2\lambda^{(A,\widehat{A}_{q})}_{k}}>\epsilon^{2} \right )\overset{p}{\longrightarrow}0,
\end{displaymath}
for any $\epsilon > 0$, where the first step follows from the triangle inequality and the last follows from the  Mar\v{c}enko-Pastur law \cite{marvcenko1967distribution}, which gives the distribution of the eigenvalues of $A^{-\frac12} \widehat{A}_q A^{-\frac12}$ and hence the generalized eigenvalues $\lambda_k^{(A,\widehat{A}_{q})}$, for sufficiently large $n$.  Alternatively, for Gaussian $y_i$, the final limit holds at any $n$ via \cite[Theorem 7]{smith2005covariance}.
Invoking \cite[Theorem 2.7]{newey1994large}, we then obtain $t_{q}^{*}\overset{p}{\longrightarrow}0$.
Since convergence in probability is preserved under continuous mappings, we also have $A_{q}^{*}\overset{p}{\longrightarrow} A$.
}
\end{proof}

\bibliographystyle{siamplain}
\bibliography{biblio}
\end{document}